\documentclass[journal,twoside,web]{ieeecolor}
\usepackage{generic}
\usepackage{amsmath}

\let\labelindent\relax

\usepackage{balance}
\usepackage{marginnote}
\usepackage{amsmath}

\let\labelindent\relax
\usepackage{textcomp}
\usepackage{amsthm}
\usepackage{amssymb}  
\usepackage{mathtools}
\usepackage[utf8]{inputenc}
\usepackage{comment}
\usepackage{algorithm,algpseudocode}
\algblock{Input}{EndInput}
\algnotext{EndInput}
\algblock{Output}{EndOutput}
\algnotext{EndOutput}

\usepackage{breqn}
\usepackage{graphicx}
\usepackage[version=4]{mhchem}
\usepackage{siunitx}
\usepackage{longtable,tabularx}
\usepackage{graphics} 
\usepackage{float}
\usepackage{epsfig} 
\usepackage{times} 
\usepackage{color}
\usepackage[dvipsnames]{xcolor}
\usepackage{pifont}
\usepackage{enumitem}
\usepackage{bm}
\usepackage{bbm}
\usepackage{soul}
\makeatletter
\newcommand{\multiline}[1]{%
  \begin{tabularx}{\dimexpr\linewidth-\ALG@thistlm}[t]{@{}X@{}}
    #1
  \end{tabularx}
}

\usepackage{arydshln}

\makeatother
\usepackage{hyperref}
\hypersetup{
    colorlinks=true,
   linkcolor=black,
    filecolor=blue,      
   urlcolor=blue,
    pdftitle={Overleaf Example},
    pdfpagemode=FullScreen,
    }

\newcommand{\trace}{\textit{tr}}

\newcommand{\E}{\mathbb{E\,}}

\newcommand{\R}{\mathbb{R}}

\newtheorem{assumption}{Assumption}

\newtheorem{lemma}{Lemma}
\newtheorem{problem}{Problem}
\newtheorem{remark}{Remark}
\newtheorem{theorem}{Theorem}

\DeclareMathOperator*{\argmin}{arg\,min}

\newcommand{\add}[1]{{\color{black} #1}} 

\newcounter {dagger} 
\newcommand{\MPar}[1]{} 

\begin{document}
\title{Data-conforming data-driven control: avoiding premature generalizations beyond data}

\author{Mohammad S. Ramadan, \IEEEmembership{Member, IEEE},\,Evan Toler,\,Mihai Anitescu, \IEEEmembership{Member, IEEE}
\thanks{The authors are with the Mathematics and Computer Science Division, Argonne National Laboratory, Lemont, IL 60439, USA,  {\tt\footnotesize mramadan@anl.gov, etoler@anl.gov, anitescu@mcs.anl.gov.}}}

\maketitle
\thispagestyle{empty}
\pagestyle{empty}

\begin{abstract}
Data-driven and adaptive control approaches face the problem of introducing sudden distributional shifts beyond the distribution of data encountered during learning. Therefore, they are prone to invalidating the very assumptions used in their own construction. This is due to the linearity of the underlying system, inherently assumed and formulated in most data-driven control approaches, which may falsely generalize the behavior of the system beyond the behavior experienced in the data. This paper seeks to mitigate these problems by enforcing consistency of the newly designed closed-loop systems with data and slowing down any distributional shifts in the joint state-input space. This is achieved through incorporating affine regularization terms and linear matrix inequality constraints to data-driven approaches, resulting in convex semi-definite programs that can be efficiently solved by standard software packages. We discuss the optimality conditions of these programs and then conclude the paper with a numerical example that further highlights the problem of premature generalization beyond data and shows the effectiveness of our proposed approaches in enhancing the safety of data-driven control methods.
\end{abstract}

\section{Introduction} \label{section: Introduction}
The development of adaptive control systems, while marked by significant theoretical advancements, has historically faced skepticism from practitioners, rooted in concerns over the reliability and robustness of such systems in real-world applications. Brian Anderson, in his article \cite{anderson2005failures}, explains reasons for this distrust and the inherent dangers of adaptive control algorithms. In this paper we highlight another fundamental vulnerability in adaptive learning and data-driven control algorithms: premature and often false generalization beyond the seen data. That is, hastily generalizing the behavior of the system beyond its behavior that was observed in the data. This is pervasive in modern data-driven control methods and can lead to catastrophic results in real-world applications. We  propose practical methods for overcoming this vulnerability in a computationally efficient manner and using standard off-the-shelf software packages.

Data has become a cornerstone across every scientific domain, and it is foundational for the advancements in artificial intelligence \cite{hey2009fourth}. In control theory, the reliance on data in control design is not new; and the subfields of system identification \cite{ljung1998system}, robust control \cite{dorato1987historical}, and adaptive control \cite{astrom2013adaptive}, pioneered decades ago, have long been central to the field of control. The historical insights these subfields have generated are still crucial to the development of modern data-driven control and reinforcement learning (RL) algorithms. Take, for example, the ``exciting'' input for the stability of adaptive control \cite{astrom1985commentary,marafioti2014persistently}, which is analogous to the role of stochastic policies in RL algorithms \cite{schulman2017proximal}, or the problem of the chaos phenomenon resulting from learning and control loops operating in comparable time scales, discovered in adaptive control \cite{mareels1988non} and later in RL \cite{wang2024fractal}. These insights have also been crucial to the development of dual control, and its automatic experiment design aspect \cite{ramadan2024extended,heirung2017dual}, which have their connections to the exploration vs exploitation trade-off in RL \cite{sutton2018reinforcement}. 

The aforementioned insights, particularly the difficulties they are coupled with, hindered the spread of adaptive control algorithms in real-world applications, and instead, practitioners and control theorists alike found refuge in the more established and guaranteed robust control algorithms \cite{safonov2012origins}. The importance of adaptive control cannot be excluded, however; and the recent surge in the application of data-driven control and RL algorithms (in power systems, for instance \cite{ekomwenrenren2023data,yuan2024reinforcement} ) and RL's role in the development of large language models \cite{bai2022training} revive this importance. Thus, we are motivated to further understand the problems of adaptive control and expand their remedies.

In this paper we highlight an extra vulnerability of adaptive and data-driven control, not addressed explicitly by Anderson in \cite{anderson2005failures}. This vulnerability is the premature and possibly false generalizations beyond data, resulting from extrapolating the behavior of the system beyond what was experienced in the data. This in turn can lead to compromising safety and performance under the existence of unmodeled nonlinearities in the true underlying system. 

Modern data-driven control methods, whether direct (model-free) \cite{coulson2019data} or indirect (model-based) \cite{willems2005note}, mostly assume the linearity of the data-generating dynamics (although a few approaches extend direct methods to very limited forms of nonlinearities \cite{de2023learning}). The problem with the linearity assumption is that it imposes the universality of the data, in the sense that the underlying dynamics must ``behave similarly'' (according to the same linear dynamics), beyond data and in any region of the state space. This assumption is often invalid, however, because various systems admit nonlinear effects beyond their engineered operating conditions and typical control design does not account for this situation.

We emphasize that in a system identification and data collection experiment, an identified system (or recorded data in the direct approaches case) is not necessarily valid under different experimental and data collection conditions. For example, a controller designed to be stabilizing/optimal with respect to the identified model (or recorded data) is  guaranteed to be stabilizing/optimal only with respect to this specific identified model. However, eventually this controller will be connected to the actual data-generating system. This connection likely will change the operating conditions of the system and therefore may invalidate the exact assumptions this controller was based on. This can possibly activate unmodeled nonlinearities that can result in instability or severely degraded performance. Therefore, even with batch learning and iterative identification in a much smaller time scale than control \cite{albertos2012iterative}, a change in the controller may still result in a rapid shift in the system's operating conditions, compromising safety and performance.

The methods we propose in this paper enforce the consistency between the designed closed-loop system and  the data encountered during learning. We achieve this consistency by augmenting the standard linear quadratic regulator (LQR) problem with affine regularization terms and linear matrix inequalities (LMIs) in the corresponding decision variables. The result is  an affine semi-definite program (SDP) with LMI constraints that can be solved, even for high-dimensional systems, efficiently and using off-the-shelf software packages. Furthermore, we introduce a single hyperparameter that enforces the consistency level; in other words, this hyperparameter is an exploration vs exploitation balance factor. Our method allows for an iterative identification (data collection) and control that prevents rapid distributional shifts and the related sudden violations of data consistency. 

In mathematical terms, this augmentation of the standard LQR problem includes working with the parametrization given by the controllability-type Gramian \cite{kailath1980linear}, which is also the steady-state covariance matrix of the system state. This makes it possible to relate the closed-loop state-input distribution to that of the data, using LMI constraints and/or regularization terms such as the Frobenius norm on covariance matrices or the Kullback--Liebler divergence between distributions. Throughout the paper we call a control design approach \textit{data-conforming} if it enforces the consistency with data according to our aforementioned methods.

Inspired by the problems addressed in \cite{anderson2005failures}, our problem is closest in spirit to the so-called Rohrs' counterexample \cite{rohrs1985robustness}, which sought to invalidate adaptive control algorithms in the 1980s. Part of Rohrs' argument is that every physical system has unmodeled (parasitic) high-frequency dynamics. Because contemporary adaptive control led to highly nonlinear loops even for simple linear systems, these loops could internally generate high-frequency signals, exciting the unmodeled dynamics and possibly leading to instability. Our concern regarding modern data-driven control is analogous to that of Rohrs' with adaptive control, but we target unmodeled nonlinear dynamics instead of unmodeled high-frequency ones. Through iterative identification and carefully expanding the closed-loop bandwidth, one can avoid exciting high-frequency dynamics \cite{anderson2002windsurfing}. Analogously, through iterative identification and careful data-conforming control design, one can avoid sudden exposures to new areas of the state space,  which potentially contain harmful, unmodeled nonlinearities.

Our data-conforming control design approach also bears resemblance to the unfalsified control approach \cite{cabral2004unfalsified} in enforcing  consistency with past data. Contrasting with this approach, which categorizes consistency through Boolean values, our approach is built around consistency in distances between distributions, allowing for generalized formulations that can cope with stochastic dynamics. Moreover, the method of \cite{cabral2004unfalsified} relies on model reference adaptive control, whereas our approach augments modern optimal control design via extending the standard LQR problem with affine regularization terms and LMI constraints. Our approach readily integrates with many modern control design formulations, allows for multivariate state-space systems, and can be extended to handle receding-horizon and input-output formulations.

The idea of consistency with data is also fundamental in the field of offline RL \cite{agarwal2020optimistic}, but there the main motivation is data efficiency by initially limiting exploration. Offline RL algorithms are generally complex \cite{fujimoto2021minimalist} and rely on stochastic nonlinear programming methods. Instead, our approach insists on computationally efficient methods that blend with modern control design techniques.

One can argue that modern data-driven policy gradient methods \cite{zhao2024data,fazel2018global} with small enough step sizes can prevent sudden distributional shifts and therefore enhance consistency with data. This is indeed a plausible argument. However, it is difficult to relate between the control gain step size and the corresponding distributional shift beyond a simple ($\epsilon-\delta$) continuity argument. Instead, our approach deals explicitly with distributions and directly dampens their shifts. Moreover, a very dampened policy gradient procedure may be slow and hence inadequate for many time-varying dynamics, limiting its applicability to constant or pseudo-constant dynamics. On the other hand, our approach decouples distributional shifts from learning system variations and therefore can accomplish both tasks effectively and simultaneously.

We conclude the paper with a simple yet telling numerical example that explains the premature generalization problem inherent in modern data-driven approaches and shows how our proposed solutions can mitigate this problem.

\section{Problem Formulation} \label{section: Problem Formulation}
Consider the dynamic system
\begin{align}
x_{k+1}&=f(x_k, u_k, w_k), \label{eq:stateEquation}
\end{align}
where $x_k\in\mathbb R^{r_x}$ is the state and $u_k\in\mathbb R^{r_u}$ is the control input. The function $f$ is bounded but unknown and possibly nonlinear. The exogenous disturbance $w_k\in\mathbb R^{r_x}$ is independent and identically distributed \add{with\MPar{8a} zero mean and covariance $W$}. It is also independent from $x_0$, the initial condition, which is also random and has finite mean and covariance.

The goal of this paper is to design a control law that minimizes the cost function \MPar{5a}
\begin{equation}
\begin{aligned}
    J &= \add{\lim_{T \to \infty} \E \left \{ \frac{1}{T} \sum_{k=0}^T \left \{ x_k^\top Q x_k + u_k^\top R u_k \right \} \right \}},\\
    &=\lim_{k \to \infty} \E \left \{x_k^\top Q x_k + u_k^\top R u_k \right \}, \label{eq:costFunction}
\end{aligned}
\end{equation}
where $Q \succeq 0$ and $R \succ 0$ are positive semi-definite and positive definite, respectively\footnote{In this paper all positive semi-definite and positive definite matrices are also symmetric.}, and the expectation averages over all the possible realizations of $x_0$ and $w_k$. \add{The\MPar{5b} two descriptions of the cost above are equivalent assuming $x_k$ and $u_k$ are bounded in their first two moments and are stationary processes as $k \to \infty$ (See the Appendix~\ref{Appendix:cost descriptions: sum and limit} for more details).}

Since $f$ is unknown, we assume that we have access to it only through experiments. Hence, we rely on the data-driven paradigm for the control design. During the first data collection experiment, we assume an initial control law 
\begin{align} \label{eq:initial control law form}
    u_k = \kappa_0 (x_k) + v_k,
\end{align}
that is locally stabilizing (not necessarily optimal) and $v_k$ is \add{a\MPar{8b} zero mean} persistently exciting  (PE) signal, in the following sense.

\begin{assumption}\label{assumption:PE}
    (The PE assumption \cite[condition~(6)]{de2019formulas}). The signal $v_k$ in \eqref{eq:initial control law form} is PE for \eqref{eq:stateEquation}. That is, given a natural number\footnote{This lower bound is the minimum number of measurements required to identify a linear model of the system \cite{willems2005note}.} $N \geq (r_u +1) r_x + r_u$, every control realization $\{u_{k}\}_{k=0}^{N}$ and the corresponding state realization $\{x_{k}\}_{k=0}^{N}$ produce the matrix
\begin{equation} \label{eq:data matrix}
     {D} := 
    \begin{bmatrix}
          X\\
          U
    \end{bmatrix}
    =
    \begin{bmatrix}
        x_0 & \hdots &x_{N}\\
        u_0 & \hdots &u_{N}
    \end{bmatrix}
\end{equation}
with full row rank $r_x + r_u$.
\end{assumption}

The PE assumption is typical in system identification and data-driven control design and equates to an effective experiment design and a minimal descriptive state definition. \add{It\MPar{6a} is important to clarify that this PE assumption makes it feasibly to achieve unbiased consistent identification when the underlying data-generating system is minimal, linear and noise-free. In the nonlinear case, Assumption~\ref{assumption:PE} simply equates to making it possible to identify a minimal linear model from data. The validity of this linear model depends on the complexity and averaged behavior of the underlying data-generating system.}

\begin{remark} \label{remark:kappa_0}
\add{The\MPar{4} feedback control law $\kappa_0$ can be linear or nonlinear, and $u_k$ can possibly include some measured human inputs and interactions (e.g. a driver or a pilot) in the experimental and data-collection phases. Whether $\kappa_0$ is locally stabilizing or the system is bounded-input, bounded-output (BIBO) stable ($\kappa_0$ is possibly zero and $v_k$'s covariance is small enough), we assume a safe experimental data collection phase that results in bounded data. Without strong restrictions on the structure of the underlying dynamics, or allowing for some aggressive and possibly destructive testing, the above assumptions are inevitable and hence are adopted in various adaptive control, data-driven control, reinforcement learning and system identification  frameworks \cite{schon2011system,ljung1998system,fazel2018global,coulson2019data}.}
\end{remark}

\begin{assumption} \label{assumption:Gaussianity}
The data \eqref{eq:data matrix} is centered around zero, \add{can be sufficiently fitted by an approximate linear model,} and the state-input joint distribution resembles a multivariate Gaussian with a positive definite covariance. 
\end{assumption}

\subsubsection*{Problem statement} Starting from the initial control law of the form \eqref{eq:initial control law form} and given access to the resulting input and state data \eqref{eq:data matrix}, with the possibility of iteratively changing the control law and collecting new data, design a linear state feedback gain $K$ that minimizes the cost $J$ in \eqref{eq:costFunction}.

Assumptions~\ref{assumption:PE} and \ref{assumption:Gaussianity} are design assumptions that make it possible to meet the problem statement aforementioned with computationally efficient solutions (building on the imposed linearity and convexity). These assumptions are true if the underlying system is linear with minimal state description. However, since we are specifically targeting nonlinear systems in this paper, we require that the data generated by the underlying nonlinear system can be approximately generated by a linear model, when each is under the same control law in closed-loop. This flexibility is to address the possibility that altering the operational conditions may activate different nonlinearities, demanding different model approximations to compensate. Although the existence of an accurate linear approximation under a certain control law is generally not guaranteed for general nonlinear systems, our assumptions are still weaker than the small signal model assumption \cite{coulson2019data} or adopting one fixed approximate linear model as in various data-driven control methods.

\add{One\MPar{12b} may ask about the purpose of considering a general nonlinear system in \eqref{eq:stateEquation} if the data can be predicted by an approximate linear model. The reason is to remember that once a new controller is designed, the resulting closed-loop system may operate in regions of the state-input space not adequately explored by the learning data and over which this approximate linear model is no longer a valid approximation. This point can be easily forgotten if we start with a linear system in \eqref{eq:stateEquation}. Therefore, the nonlinearity of \eqref{eq:stateEquation} is a reminder that even if $f$ can be approximated by a linear model under certain conditions, this approximation may not be extrapolated to different conditions.}

\section{Background} \label{section: Background}
In this section we start by presenting the standard LQR problem. Although the problem formulation targets the nonlinear dynamics \eqref{eq:stateEquation}, our derivations in the subsequent sections assume the existence of a valid linear approximation under each operating condition.

\subsection{Standard LQR}
Suppose, in this subsection only, that the state-space system \eqref{eq:stateEquation} is linear,
\begin{align*}
x_{k+1} = f(x_k,u_k,w_k) = A x_k + B u_k + w_k,
\end{align*}
where $x_0, w_k$ are random variables and their statistics are as in \eqref{eq:stateEquation}, and the pair $(A,B)$ is controllable \cite{kailath1980linear}. The LQR problem is  obtaining a controller of the form $u_k = K x_k + v_k$, where $v_k$ is a white noise of zero mean and a fixed user-defined covariance $V \succ 0$, to minimize \eqref{eq:costFunction}. The signal $v_k$ is to preserve Assumption~\ref{assumption:PE} for future or iterative control design.

The cost function \eqref{eq:costFunction}, up to an additive constant in $K$, can be rewritten as (for a detailed derivation, check Appendix~\ref{Appendix:P parametrization cost})
\begin{align} \label{eq:cost Observability type}
    J = \trace \left ( P \left [ W + BV B^\top \right ] \right ),
\end{align}
where $P$ is the observability-type Gramian \cite{kailath1980linear}, given by
\begin{align}
    P = \sum_{k=0}^\infty  [A+BK]^{k\,\top} \left [ Q + K^\top R K \right ] [A+BK]^k,
    \label{eq:Observability Gramian}
\end{align}
where $Q \succeq 0$ and $R \succ 0$ are the user-defined weighting matrices, as in \eqref{eq:costFunction}, and $(Q,A)$ is assumed detectable. The optimal control gain $K$ (which minimizes $J$) is then given by the solution to the following.
\begin{quote}
\itshape
Standard LQR (Observability-type Gramian) \cite[Sec.~4.1]{bertsekas2012dynamic}:
\begin{align*}
        \text{Evaluate }K = -\left [R+B^\top P B \right ]^{-1} B^\top P A,
\end{align*}
when $P$ satisfies the algebraic Riccati equation
\begin{align*}
     0 = Q + A^\top P A - P -A^\top P B \left [ R + B^\top P B \right]^{-1} B^\top P A.
\end{align*}
\end{quote}

The problem above results from the square completion in $K$ of the Lyapunov equation
\begin{align}
     0 = [A+BK]^{\top} P [A+BK] - P + Q + K^\top R K, \label{eq:Lyap observability P}
\end{align}
then chosing $K = -\left [R+B^\top P B \right ]^{-1} B^\top P A$, which corresponds to the minimum $P$ (in the Loewner ($\succeq$) ordering sense).

There is an equivalent characterization to solving the LQR problem. Using the linearity and the cyclic property of the trace, the cost can be written as (check Appendix~\ref{Appendix:P parametrization cost} for more details)
\begin{equation}\label{eq:cost Controllability type}
\begin{aligned} 
    J &= \trace \left (\left [ Q + K^\top R K \right ] \Sigma  \right ),\\
    &= \trace \left ( Q \Sigma  \right )  +\trace \left ( R K \Sigma  K^\top \right ),
\end{aligned}
\end{equation}
where $\Sigma $ is the controllability-type Gramian, given by
\begin{align}
    \Sigma  &= \lim_{k \to \infty}\E \left \{x_k x_k^\top \right \}, \label{eq:Sigma as state covariance}\\
    &=\sum_{k=0}^\infty \left [ A+BK\right ]^k\left [ W + BV B^\top \right ]\left [ A+BK\right ]^{k\,\top}, \nonumber
\end{align}
and is also the solution of the Lyapunov equation
\begin{align}
    \Sigma = \left [ A+BK\right ] \Sigma \left [ A+BK\right ]^\top + W + BV B^\top. \label{eq:Controllability Lyapunov}
\end{align}
Since $V \succ 0$ and the pair $(A,B)$ is controllable, the controllability Gramian satisfies $\Sigma \succ 0$.

Using this parametrization, we first define the extra variable $Z_0$, such that $Z_0 \succeq K \Sigma  K^\top$, which is equivalent to (the Schur complement corresponding to) the LMI\footnote{Since congruence preserves definiteness, that is, for any full rank matrix $T$, $T \Sigma T^\top \succ 0$ if and only if $\Sigma \succ 0$. $T$ is taken to be $T = \begin{bmatrix}
I & -K \\ 
0 & I
\end{bmatrix}$ in the above instance, and hence, both $Z_0 - K \Sigma  K^\top$ and $\Sigma$ appear on the block diagonal of $T \Sigma T^\top$. Since $\Sigma \succ0 $, it must be true that $Z_0 - K \Sigma  K^\top \succeq 0$.}
\begin{align*}
    \begin{bmatrix}
        Z_0 & K \Sigma\\
        \Sigma K^\top & \Sigma
    \end{bmatrix} \succeq 0 , \text{ or, }
    \begin{bmatrix}
        Z_0 & L\\
        L^\top & \Sigma
    \end{bmatrix} \succeq 0,
\end{align*}
using the change of variables $L = K \Sigma$. The stability-imposing inequality relaxation of \eqref{eq:Controllability Lyapunov} (we discuss the consequences of this relaxation in Section~\ref{section:optimality conditions}) can also be described as an LMI, that is,
\begin{align*}
    &\Sigma \succeq \left [ A+BK\right ] \Sigma \left [ A+BK\right ]^\top + W + BV B^\top \\
    &\iff \Sigma \succeq \left [ A\Sigma+BL\right ] \Sigma^{-1} \left [ A\Sigma+BL\right ]^\top + W + BV B^\top \\
    &\iff 
    \begin{bmatrix}
        \Sigma-W-B V B^\top & A \Sigma + B L \\
        \Sigma A^\top + L^\top B^\top & \Sigma
    \end{bmatrix}\succeq 0.
\end{align*}
Therefore, the LQR problem can be described as an SDP with affine cost and LMI constraints \cite{boyd1993control}.
\begin{quote}
\itshape
Standard LQR (Controllability-type Gramian):\footnote{From a numerical perspective, to avoid open sets in the domain of the optimization problem, $\Sigma \succ 0$ can be replaced by $\Sigma \succeq \epsilon_0 I$, where $\epsilon_0>0$ is very small.}
\begin{equation*}
    \begin{aligned}
        &\min_{\Sigma ,\,L,\,Z_0} \trace \left (  Q\Sigma  \right ) + \trace \left ( R Z_0 \right) \\
        &\text{s.t. }\Sigma \succ 0, \quad
        \begin{bmatrix}
            Z_0 & L\\
            L^\top & \Sigma 
        \end{bmatrix}\succeq 0,\\
        &\begin{bmatrix}
        \Sigma-W-B V B^\top & A \Sigma + B L \\
        \Sigma A^\top + L^\top B^\top & \Sigma
    \end{bmatrix}\succeq 0,
    \end{aligned} 
\end{equation*}
where the change of variables $L=K \Sigma $ has been used and the optimal control is then recovered from the optimal values $\Sigma_\star$ and $L_\star$ through $K_\star = L_\star \Sigma_\star^{-1}$.
\end{quote}

\add{The\MPar{9a} computational complexity of the above SDP (and all the SDPs in this paper) is of $\mathcal{O}(r_x^3)$, where $r_x$ is the state-dimension. This is similar to the computational complexity of solving the corresponding algebraic Riccati equation, as it involves matrix decompositions and inversions typically of $\mathcal{O}(r_x^3)$ as well \cite{boyd1994linear}.}

The choice of using the controllability-type parametrization of the cost \eqref{eq:cost Controllability type} is motivated by \eqref{eq:Sigma as state covariance}; the matrix $\Sigma $ is the steady-state state covariance matrix and thus has a statistical interpretation. Since this paper is in the context of data-driven control, the matrix $\Sigma$ is of great importance, and having it appear explicitly is a key. We show this in Section~\ref{section: infinite-horizon K}.

\subsection{Data-driven certainty equivalence}
In the case of unknown system matrices $A$ and $B$, conventional (indirect) data-driven approaches identify a model $(\widehat A,\widehat B)$ and $\widehat W$, then solve the standard LQR problem with the identified model assumed to represent the true system exactly, a procedure termed certainty equivalence LQR.

The estimate model $(\widehat A,\widehat B)$ and $\widehat W$ can be obtained as a solution to the least-squares problem \cite{zhao2024data}
\begin{align} \label{eq:linear state-space ID}
\begin{aligned}
[ \widehat A, \widehat B] &= \argmin_{ A, B} \left \lVert   \widehat X_{err}
    \right \rVert_F,\\
    \widehat W &= \frac{1}{N} \widehat X_{err} \widehat X_{err}^\top,\\
    \widehat X_{err} &=X_1 - [A,  B] 
    \begin{bmatrix}
          X_0 \\
          U_0
    \end{bmatrix},
     \end{aligned}
\end{align}
where $\lVert \cdot \rVert_F$ is the Frobenius norm. Using the acquired data \eqref{eq:data matrix}, we let $  X_0:=[x_0,\hdots,x_{N-1}]$, $  U_0:=[u_0,\hdots,u_{N-1}]$, and $  X_1:=[x_{1},\hdots,x_{N}]$.

The certainty equivalence LQR problem can be decomposed into an identification problem and then a control design problem. We describe this two-level approach as the following problem.

\begin{problem} \label{prob:CE-LQR}
Certainty equivalence LQR:
\begin{equation*}
    \begin{aligned}
        &\min_{\Sigma ,\,L,\,Z_0} \trace \left (  Q\Sigma  \right ) + \trace \left ( R Z_0 \right) \\
        &\text{s.t. }\Sigma \succ 0, \quad
        \begin{bmatrix}
            Z_0 & L\\
            L^\top & \Sigma 
        \end{bmatrix}\succeq 0,\\
        &\begin{bmatrix}
        \Sigma-\widehat W-\widehat B V \widehat B^\top & \widehat A \Sigma + \widehat B L \\
        \Sigma \widehat A^\top + L^\top \widehat B^\top & \Sigma
    \end{bmatrix}\succeq 0,\\
        & \widehat A, \widehat B, \widehat W \text{ as in \eqref{eq:linear state-space ID}}.
    \end{aligned} 
\end{equation*}
The optimal control is recovered from the optimal values $\Sigma_\star$ and $L_\star$ through $K_\star = L_\star \Sigma_\star ^{-1}$.
\end{problem}

This data-driven approach is a basic certainty equivalence control design approach. For a discussion of the robustness and the sample efficiency, and to see how to enforce robustness under the existence of noise ($w_k \neq 0$) when the underlying model \eqref{eq:stateEquation} is linear, an interested reader can consult \cite{dean2020sample} and \cite{van2021matrix}. In this paper, however, we tackle the problem of the false generalization of nonlinear systems by linear ones, beyond data, which is implicit in data-driven control algorithms.

Notice if the underlying system is indeed linear, the region of the state-input space from which the data \eqref{eq:data matrix} has been acquired is irrelevant, as long as the PE assumption, Assumption~\ref{assumption:PE}, is met. The reason is that if the system is linear, predicting its behavior in any region equates to predicting its behavior universally across the state-input space. This generally does not hold if the underlying system is nonlinear.

One shortcoming of using Problem~\ref{prob:CE-LQR} as the optimal control model is that once a linear model is identified and the resulting gain matrix $K$ is applied in feedback, the region in the state space, where the new closed-loop system is operating, may not be within or close to the region explored by the data \eqref{eq:data matrix} in the identification step. Therefore, if the underlying system is nonlinear\footnote{This is typically the realistic case, as many systems that are considered linear are not so beyond some operating conditions.}, the new controller may force the system to go to regions of the domain of $f$ over which the system admits behaviors that are very different from what was experienced in the data. This implicit and inherent false universal generalization may in turn result in serious violations of safety or lead to a degraded performance.

\section{Data-conforming data-driven control} \label{section: infinite-horizon K}

In practice, for a controller to conform to the learning data, a measure of similarity between distributions has to be adopted and augmented to the cost/constraints of the LQR problem. For this purpose, several approaches and simplifications can be considered. We explore some of them in this section.

\subsection{Conforming to the state data}
Notice that Problem~\ref{prob:CE-LQR} is a convex optimization problem and that the decision variable $\Sigma $, the steady-state state covariance matrix, can be enforced to be similar to the state empirical covariance from data, via convex constraints or regularization.

Using \eqref{eq:data matrix}, let $\Sigma_{data}$ be the state empirical covariance:
\begin{align}
    \Sigma_{data} = \frac{1}{N+1}  {X} {X}^\top, \label{eq:Sigma_data}
\end{align}
\add{which\MPar{3a} is the maximum likelihood estimate of $\Sigma$ given $X$. $\Sigma_{data}$ is unbiased and follows a Wishart distribution \cite[Ch.~7]{anderson1958introduction}. The generalized variance, that is $\det \Sigma_{data}$, is asymptotically normal with mean $\det \Sigma$ and variance $2 r_x \left (\det \Sigma \right)^2 / N$ \cite[Thm.~7.5.4]{anderson1958introduction}.}

The data-conforming (in the state) version of Problem~\ref{prob:CE-LQR} can now be stated as follows.

\begin{problem} \label{prob:CE-LQR-state-conforming-hard}
Data-conforming (in the state) LQR via a hard constraint:
\begin{equation*}
    \begin{aligned}
        &\min_{\Sigma ,\,L,\,Z_0} \trace \left (  Q\Sigma  \right ) + \trace \left ( R Z_0 \right) \\
        &\text{s.t. }\Sigma \succ 0,  \quad
        \begin{bmatrix}
            Z_0 & L\\
            L^\top & \Sigma 
        \end{bmatrix}\succeq 0, \quad
        \Sigma  = \Sigma_{data},\\
        &\begin{bmatrix}
        \Sigma-\widehat W-\widehat B V \widehat B^\top & \widehat A \Sigma + \widehat B L \\
        \Sigma \widehat A^\top + L^\top \widehat B^\top & \Sigma
    \end{bmatrix}\succeq 0,\\
        & \widehat A, \widehat B, \widehat W \text{ as in \eqref{eq:linear state-space ID}},
    \end{aligned} 
\end{equation*}
and, similarly, the optimal control is recovered from the optimal values $\Sigma_\star $ and $L_\star$ through $K_\star = L_\star \Sigma_\star ^{-1}$.
\end{problem}

The new hard linear constraint $\Sigma  = \Sigma_{data}$ has two shortcomings:  (i) possible numerical instabilities and/or feasibility issues and (ii) limits to exploration beyond learning data. That is, potential future improvements via some exploration is not possible because the new controller will always seek to generate the same state distribution. This can be relaxed by using some margin, say, $\Sigma  \succeq \Sigma_{data} - \epsilon \mathbb I$ and $\Sigma  \preceq \Sigma_{data} + \epsilon \mathbb I$, for some $\epsilon>0$, predetermined or possibly included in the cost as a slack variable. This flexibility in the constraint, resembled 
in $\epsilon$, can determine the exploration vs exploitation balance of the new control design.

One can also design a convex regularization term using some norm on the covariance matrices. For example, the squared Frobenius norm $\lVert \Sigma  - \Sigma_{data} \rVert_F^2 = \trace \left ( \left[\Sigma  - \Sigma_{data} \right]\left [\Sigma  - \Sigma_{data} \right]^\top\right)$ can be used by minimizing $\trace \left( Z'\right)$, where $Z' \succeq \left[\Sigma  - \Sigma_{data} \right]\left [\Sigma  - \Sigma_{data} \right]^\top$, or equivalently as an LMI \cite{boyd1994linear},
\begin{align*}
\begin{bmatrix}
        Z' & \Sigma  - \Sigma_{data}\\
        \Sigma  - \Sigma_{data}& I
    \end{bmatrix} \succeq 0.
\end{align*}
We use this regularization term in the following modified problem.

\begin{problem} \label{prob:CE-LQR-state-conforming-regularization}
Data-conforming (in the state) LQR via regularization:
\begin{equation*}
    \begin{aligned}
        &\min_{\Sigma ,\,L,\,Z_0,\,Z'} \trace \left (  Q\Sigma  \right ) + \trace \left ( R Z_0 \right) + \gamma' \trace \left (Z' \right ) \\
        &\text{s.t. }\Sigma \succ 0, \quad
        \begin{bmatrix}
            Z_0 & L\\
            L^\top & \Sigma 
        \end{bmatrix}\succeq 0,\\
        &\begin{bmatrix}
        \Sigma-\widehat W-\widehat B V \widehat B^\top & \widehat A \Sigma + \widehat B L \\
        \Sigma \widehat A^\top + L^\top \widehat B^\top & \Sigma
    \end{bmatrix}\succeq 0,\\
        &\begin{bmatrix}
            Z' & \Sigma  - \Sigma_{data}\\
            \Sigma  - \Sigma_{data}& I
        \end{bmatrix} \succeq 0,\\
        & \widehat A, \widehat B, \widehat W \text{ as in \eqref{eq:linear state-space ID}},
    \end{aligned} 
\end{equation*}
where $\gamma'>0$ and the optimal control is recovered from the optimal values $\Sigma_\star $ and $L_\star$ through $K_\star = L_\star \Sigma_\star ^{-1}$.
\end{problem}

The user-defined weight $\gamma'$  plays an analogous role to $\epsilon$ in the discussion succeeding Problem~\ref{prob:CE-LQR-state-conforming-hard}, in determining the exploration vs exploitation balance of the new control design.

\add{Note\MPar{1} that our problem formulation, in particular, equation \eqref{eq:stateEquation}, assumes that the state is fully observed; therefore, the identification step \eqref{eq:linear state-space ID} is a straightforward least-squares problem. For partial state observation, a state-space model can be estimated, for example, using an eigensystem realization algorithm \cite{peterson1995efficient}. In such cases, a state observer (such as a Kalman filter) can be used, and Problem~\ref{prob:CE-LQR-state-conforming-regularization} (along with all optimization problems in this paper) can be reformulated based on the state estimate provided by the observer. Alternatively, one could bypass the identification step and reformulate the data-conforming control under the behavioral data-driven design philosophy \cite{willems2005note,coulson2019data}, as demonstrated in \cite{ramadan2025floodgates}.}

Only the state data distribution is used in Problems~\ref{prob:CE-LQR-state-conforming-hard} and \ref{prob:CE-LQR-state-conforming-regularization}. This puts no emphasis on the input data distribution and carries the implicit assumption that the underlying true system \eqref{eq:stateEquation} is linear with constant coefficients in $u_k$. That is, the seen inputs have the same effect regardless of the system state at their occurrence. In the following subsection, however, we enforce data conformation in the state-input joint distribution instead, which is more suited to general nonlinear systems.

\subsection{Conforming to the joint state-input data}
We denote the steady-state state-input joint density of the new design and the one estimated from the data by\footnote{The notation $ \mathcal{N}(\mu', \Sigma')$ denotes a Gaussian density of mean $\mu'$ and covariance $\Sigma'$.} $ \mathcal{N}_{des} =  \mathcal{N}(\mu_{des}, \Gamma_{des})$ and $\mathcal{N}_{data} =  \mathcal{N}(\mu_{data}, \Gamma_{data})$, respectively. The Gaussian property of these densities and $\mu_{des}=\mu_{data}=0$ follow from Assumption~\ref{assumption:Gaussianity}. Therefore, these densities are fully characterized by their covariance matrices.

Since $u_k = K x_k + v_k$, the design covariance matrix satisfies
\begin{align} \label{eq:design covariance}
    \Gamma_{des} &= 
    \begin{bmatrix}
        \lim_{k \to \infty} \E x_k x_k^\top & \lim_{k \to \infty} \E x_k u_k^\top \\
        \lim_{k \to \infty} \E  u_k x_k^\top & \lim_{k \to \infty} \E  u_k u_k^\top
    \end{bmatrix}\nonumber\\
    &= 
    \begin{bmatrix}
        \Sigma & \Sigma K^\top \\
        K\Sigma & K\Sigma K^\top + V
    \end{bmatrix},
\end{align}
where $\Sigma$ is as defined in \eqref{eq:Sigma as state covariance}. The empirical covariance matrix satisfies
\begin{align}
\begin{aligned} \label{eq:Gamma data def}
    \Gamma_{data}&\approx \frac{1}{N+1}  \left [D -\mu_{data} \right ]  \left [ D-\mu_{data} \right ] ^\top \\
    &= 
    \begin{bmatrix}
         \Sigma_{data} &  H_{data}\\
         H_{data}^\top &  M_{data}
    \end{bmatrix},
\end{aligned}
\end{align}
where $H_{data} = (N+1)^{-1} X U^\top$ and $M_{data} = (N+1)^{-1} U U^\top$. \add{Similar\MPar{3b} to $\Sigma_{data}$, $\Gamma_{data}$ is also unbiased, follows a Wishart distribution, and has an asymptotically normal generalized variance \cite[Ch.~7]{anderson1958introduction}.}

It is not clear whether one can enforce the closeness of $\Gamma_{data}$ to $\Gamma_{des}$, as was done in Problems~\ref{prob:CE-LQR-state-conforming-hard} and \ref{prob:CE-LQR-state-conforming-regularization}, via LMIs in $\Sigma$ and $K$ (or $L$). Instead, we use the Kullback--Leibler  divergence, notated $d_{KL}( \cdot \mid \mid \cdot )$, between the densities $ \mathcal{N}_{des}$ and $ \mathcal{N}_{data}$, which reduces to an expression in terms of their covariances. This expression, under some relaxation, can be posed as an affine regularization term and LMI constraints in $\Sigma$ and $L$.

Although the KL divergence is not a full-fledged metric, it is non-negative, and zero if and only if the two distributions are equal (almost everywhere) \cite[Sec. 8.6]{thomas2006elements}. It is a measure of the inefficiency of assuming one distribution when the true distribution is another, and it has been used in the contexts of exploration vs. exploitation in artificial intelligence and RL \cite{schulman2017proximal} and in constructing ambiguity sets in distributionally robust optimization techniques \cite{mohajerin2018data}. 

\begin{lemma} \cite{pardo2018statistical}
    The Kullback--Leibler divergence metric between $ \mathcal{N}_{des}$ and $ \mathcal{N}_{data}$ is
    \begin{align}
        &d_{KL}(  \mathcal{N}_{des} \mid \mid  \mathcal{N}_{data}) = \int_{\R^{r_x + r_u}}  \mathcal{N}_{des}(x) \log \frac{  \mathcal N_{des}(x)}{  \mathcal N_{data}(x)} dx \nonumber\\
        &= \frac{1}{2} \Bigg [ \trace (\Gamma_{data}^{-1}\Gamma_{des}) + (\mu_{des} - \mu_{data})^\top \Gamma_{data}^{-1}(\mu_{des} - \mu_{data})\nonumber\\
        &\hskip 15mm - (r_x + r_u) + \log \frac{\text{det }\Gamma_{data}}{\text{det }\Gamma_{des}}\Bigg ]. \label{eq:KL divergence of Gaussians}
    \end{align} \hfill \qed
\end{lemma}

We drop the second term in \eqref{eq:KL divergence of Gaussians}, since $\mu_{des}=\mu_{data}=0$, and the additive constant $-(r_x + r_u)$, since it does not alter the optimal control. The modified KL-based regularization term is then
\begin{align*}
    \bar d_{KL}(  \mathcal{N}_{des} \mid \mid  \mathcal{N}_{data}) = \frac{1}{2} \Bigg [ \trace (\Gamma_{data}^{-1}\Gamma_{des}) + \log \frac{\text{det }\Gamma_{data}}{\text{det }\Gamma_{des}}\Bigg ].
\end{align*}

We now employ an approximation to reshape $\bar d_{KL}(  \mathcal{N}_{des} \mid \mid  \mathcal{N}_{data})$ into a form amenable to the linear control design machinery, using the following results.

\begin{lemma}\label{lemma:log det to tr log}
The algebraic equality
\begin{align} \label{eq:log det to tr log}
    \log \frac{\text{det }\Gamma_{data}}{\text{det }\Gamma_{des}} = \trace \log \left (  \Gamma_{data}\Gamma_{des}^{-1} \right)
\end{align}
holds.\footnote{The function $\log$ to the left is the real-valued logarithmic function while $\log$ to the right is the logarithm of a matrix, in the sense that a matrix exponential of the $\log$ of a matrix equals the matrix itself.}
\end{lemma}
\begin{proof}
This is immediate from a result in linear algebra for positive definite matrices. We can write
\begin{align*}
    \log \frac{\text{det }\Gamma_{data}}{\text{det }\Gamma_{des}} 
    &= \log \text{det }\Gamma_{data} - \log \text{det }\Gamma_{des},\\
    &=^{(1)} \log \text{det }\Gamma_{data} + \log \text{det }\Gamma_{des}^{-1},\\
    &=^{(2)} \trace \log \Gamma_{data} + \trace \log \Gamma_{des}^{-1},\\
    &=^{(3)} \trace \log \left ( \Gamma_{data}\Gamma_{des}^{-1} \right ),
\end{align*}
where $(1),(2),(3)$ are implied by the positive definiteness of $\Gamma_{des}$ and $\Gamma_{data}$, and the identity $\log \text{det}\, \cdot = \trace \log \cdot$ for nonsingular matrices.
\end{proof}

Toward an approximation of \eqref{eq:log det to tr log} that is convex in $\Gamma_{des}$, suppose $\Gamma_{data}\Gamma_{des}^{-1} \approx I$. That is, suppose $\Gamma_{data}\Gamma_{des}^{-1} = I + \Theta$ for a small perturbation matrix $\Theta$. The matrix logarithm can be written in terms of its Taylor expansion as
\begin{align*}
    \log \left ( I + \Theta \right)= \Theta - \frac{1}{2} \Theta^2 + \frac{1}{3} \Theta^{3} + \hdots,
\end{align*}
which, for a first-order truncation, satisfies
\begin{align}
\log \left ( \Gamma_{data}\Gamma_{des}^{-1} \right) = 
    \log \left ( I + \Theta \right) &\approx \Theta  = \Gamma_{data}\Gamma_{des}^{-1} - I. \label{eq:approx log}
\end{align}

Substituting the approximation \eqref{eq:approx log}, after dropping the identity (additive constant), in \eqref{eq:log det to tr log}, we can use a surrogate version of $\bar d_{KL}$, call it $F$,  as a regularization term. For a fixed $\Gamma_{data}$,
\begin{align}
    F(\Gamma_{des}) = \trace \left ( \Gamma_{data}^{-1}\Gamma_{des} + \Gamma_{data}\Gamma_{des}^{-1}\right),
\end{align}
where the $1/2$ factor in $\bar d_{KL}$ is dropped and will be replaced by a hyperparameter $\gamma>0$. The new version of the cost \eqref{eq:cost Controllability type}, with $F$ as a regularization term and $\gamma>0$ a user-defined hyperparameter, is now designated by
\begin{align}\label{eq:cost KL Div editted}
     J_{F} &= \trace \left ( Q \Sigma  \right )  +\trace \left ( R K \Sigma  K^\top \right ) +  \gamma F.
\end{align}

Next we show that the regularization term $F$, although an approximation of the KL divergence, has favorable properties for our purposes.

\subsection{Properties of $F$}
\add{We show that $\Gamma_{des} = \Gamma_{data}$ is the unique minimizer of the regularization term $F$, or equivalently of $J_{F}$ as $\gamma \to \infty$. That is, even after the approximation \eqref{eq:approx log}, the resulting regularization term $F$ still enforces $\Gamma_{des}$ to be close to $\Gamma_{data}$.\MPar{5}}

We take as given the arithmetic and geometric means inequality, which we state as the following lemma.
\begin{lemma} \label{lemma:lambda + 1/lambda}
Let $\lambda > 0$. Then $\lambda + 1/\lambda \geq 2$, with equality if and only if $\lambda=1$. \hfill \qed
\end{lemma}
From this, we conclude the subsequent lemma about positive definite matrices.
\begin{lemma} \label{lemma:trace of E E inverse}
Let $E$ be an $n \times n$ symmetric positive definite matrix. Then $E = \mathbb I$ is the unique minimizer of $\bar F(E)=\trace \left(E + E^{-1} \right)$.
\end{lemma}
\begin{proof}
Since $E$ is symmetric positive definite, we may write the unitary diagonalization $E = \Psi \Lambda \Psi^T$, where $\Lambda \succ 0$ is diagonal and $\Psi$ is unitary. Therefore, by algebraic identities $\bar F(E)=\trace \left(\Psi \left ( \Lambda + \Lambda^{-1} \right ) \Psi^{T} \right)=\trace \left(\left ( \Lambda + \Lambda^{-1} \right ) \Psi^{T}\Psi  \right)=\trace \left(\Lambda + \Lambda^{-1}   \right)$. That is,
\begin{align*}
    \bar F(E)=\trace \left(\Lambda + \Lambda^{-1} \right) = \sum_i^n \left ( \lambda_i + \frac{1}{\lambda_i} \right ),
\end{align*}
where $\lambda_i>0,\,i=1,\hdots,n,$ are the eigenvalues of $E$. The unique minimum happens when each $\lambda_i = 1$, per Lemma~\ref{lemma:lambda + 1/lambda}. That is, when $E=\Lambda= \mathbb I$ ($\mathbb I$ is the only positive definite matrix of eigenvalues of $1$),  it is the unique minimizer of $\bar F$.
\end{proof}
\begin{theorem} \label{theorem:global minimizer}
    $\Gamma_{des} = \Gamma_{data}$ is the global minimizer of
    \begin{align*}
        F(\Gamma_{des}) = \trace \left ( \Gamma_{data}^{-1}\Gamma_{des} + \Gamma_{data}\Gamma_{des}^{-1}\right).
    \end{align*}
\end{theorem}
\begin{proof}
By the linearity and the cyclic property of the trace, we can rewrite $F$ as follows:
\begin{align*}
    F(\Gamma_{des}) = \trace \left ( \Gamma_{data}^{-\frac{1}{2}}\Gamma_{des}\Gamma_{data}^{-\frac{1}{2}} + \Gamma_{data}^{\frac{1}{2}}\Gamma_{des}^{-1}\Gamma_{data}^{\frac{1}{2}}\right).
\end{align*}
If we let $E=\Gamma_{data}^{-\frac{1}{2}}\Sigma_{des}\Sigma_{data}^{-\frac{1}{2}}$, then we have
\begin{align*}
    \bar F(E) = \trace \left ( E + E^{-1}\right).
\end{align*}
Showing that $\Gamma_{des} = \Gamma_{data}$ is the unique global minimizer of $F$ equates to showing that $E=\mathbb I$ is the unique global minimizer of $\bar F$, which is true by Lemma~\ref{lemma:trace of E E inverse}.
\end{proof}

Next we show that the regularization term $F$ not only has a unique global minimizer that we seek to enforce, that is, $\Gamma_{des} = \Gamma_{data}$, but is also computationally appealing, convex in particular. 

\begin{theorem} \label{theorem:convexity}
    $F$ is convex in $\Gamma_{des}$.
\end{theorem}
\begin{proof}
    The domain---the set of all positive definite matrices---is convex. The first term of $F(\Gamma_{des})$ is linear in $\Gamma_{des}$, while the second term
    \begin{align*}
        \trace \left (\Gamma_{data}\Gamma_{des}^{-1}\right) &= \trace \left (\Gamma_{data}^{\frac{1}{2}}\Gamma_{des}^{-1}\Gamma_{data}^{\frac{1}{2}}\right)= \sum_i \xi_i ^T \Gamma_{des}^{-1} \xi_i,
    \end{align*}
    where $\xi_i$ is the $i$th column of $\Gamma_{data}^{\frac{1}{2}}$.
    Each entry in the sum is convex in $\Gamma_{des}$, as shown in \cite[p.~76]{boyd2004convex}.
\end{proof}

Notice that $\frac{1}{2}F$ is the Jeffreys divergence, a discrepancy metric between distributions \cite{pardo2018statistical} (Gaussians in this case). The above derivations highlight its relation to the KL divergence, which is more familiar in the exploration vs exploitation context.

The above results are concerned with the regularization term as a function of $\Gamma_{des}$. We show next that, under some relaxation, $F$ can be related back to the original decision variables $\Sigma$ and $L$ ($L=K \Sigma$) affinely in the cost, together with LMI constraints.

\subsection{Representing $F$ in terms of $\Sigma$ and $L$}
The first term of $F$, $\trace \left(\Gamma_{data}^{-1}\Gamma_{des} \right)$, can be upper-bounded by $\trace \left(\Gamma_{data}^{-1}Z_1 \right)$, where $Z_1 \succeq \Gamma_{des}$, or equivalently, committing to the change of variables $L=K \Sigma$,
\begin{align*} 
    Z_1 \succeq \Gamma_{des} &= \begin{bmatrix}
        \Sigma & \Sigma K^\top \\
        K\Sigma & K\Sigma K^\top + V
    \end{bmatrix},\\
    &=
    \begin{bmatrix}
        \Sigma\\
        L
    \end{bmatrix} \Sigma^{-1}
    \begin{bmatrix}
        \Sigma & L^\top
    \end{bmatrix}+
    \mathcal{V},
\end{align*}
which can be equivalently described by the LMI
\begin{align}\label{eq:Z_1 LMI}
\begin{bmatrix}
        Z_1 - 
        \mathcal{V}
    & \begin{bmatrix}
        \Sigma\\
        L
    \end{bmatrix}\\
    \begin{bmatrix}
        \Sigma & L^\top
    \end{bmatrix} & \Sigma
\end{bmatrix} \succeq 0,
\end{align}
where $\mathcal{V}= \text{block-diag}(0_{r_x \times r_x},V)$.

The second term is more involved since it contains the inverse of $\Gamma_{des}$. Applying the inverse of a partitioned matrix \cite{horn2012matrix}, we have
\begin{align*}
\Gamma_{des}^{-1} = 
\begin{bmatrix}
        \Sigma^{-1}+K^\top V^{-1} K & -K^\top V^{-1} \\
        -V^{-1} K & V^{-1}
\end{bmatrix},
\end{align*}
which exists and is positive definite, since $\Gamma_{des} \succ 0$. Hence, the second term of $F$
\begin{align} \label{eq:second term of F}
   \trace \left (\Gamma_{data} \Gamma_{des}^{-1} \right) &= \trace \left (\Sigma_{data} \Sigma^{-1} + \Sigma_{data} K^\top V^{-1} K \right ) +\nonumber\\
   &\hskip-10mm \trace \left ( M_{data} V^{-1} -2H^\top_{data} K^\top V^{-1}\right ),\nonumber\\
   &= \trace \left (\Sigma_{data} \Sigma^{-1} \right ) + \trace \left (V^{-1} K \Sigma_{data} K^\top \right ) +\nonumber\\
   &\hskip-10mm \trace \left ( M_{data} V^{-1} \right ) + \trace \left (-2K^\top V^{-1}H^\top_{data}  \right ).
\end{align}

\begin{lemma} \label{lemma: completing the squares}
 We can rewrite the following term as
 \begin{equation}
    \begin{aligned}\label{eq: completing the squares}
        &\trace \left (\Sigma_{data} K^\top V^{-1} K -2K^\top V^{-1}H^\top_{data}\right ) = \\
        &\trace \Big(V^{-1}\left[K \Sigma- H_{data}^\top \Sigma_{data}^{-1} \Sigma\right] \Sigma^{-1} \Sigma_{data} \Sigma^{-1} \times \\
        &\hskip -10mm\left[K \Sigma- H_{data}^\top \Sigma_{data}^{-1} \Sigma\right]^\top - V^{-1}H_{data}^\top \Sigma_{data}^{-1} H_{data}\Big).
    \end{aligned}
    \end{equation}
\end{lemma}
\begin{proof}
Using the linearity and the cyclic property of the trace, we reduce the right-hand side to the left one.
\end{proof}

After dropping the additive constants in $\Sigma$ and $K$, from \eqref{eq:second term of F} and \eqref{eq: completing the squares}, and toward forming LMIs, we relax the right-hand side of \eqref{eq: completing the squares} by assuming $\Sigma \approx \Sigma_{data}$; thus, $\Sigma^{-1} \Sigma_{data} \Sigma^{-1} \approx \Sigma ^ {-1}$. We use the extra variable $Z_2$ such that
\begin{align*}
    Z_2 \succeq \left[K \Sigma- H_{data}^\top \Sigma_{data}^{-1} \Sigma\right] \Sigma^{-1} \left[K \Sigma- H_{data}^\top \Sigma_{data}^{-1} \Sigma\right]^\top,
\end{align*}
which, equivalently, as an LMI, is
\begin{align} \label{eq:Z_2 LMI}
    \begin{bmatrix}
        Z_2 & L- H_{data}^\top \Sigma_{data}^{-1} \Sigma \\ 
        \left[L- H_{data}^\top \Sigma_{data}^{-1} \Sigma\right]^\top & \Sigma
    \end{bmatrix} \succeq 0.
\end{align}

We show empirically in Section~\ref{section: Numerical} that, even after the above relaxation, the consistency with state-input joint distribution is still effectively enforced.

The term $\trace \left (\Sigma_{data} \Sigma^{-1} \right )$ can also be described by an LMI by using an extra variable $Z_3 \succeq \Sigma^{-1}$. Hence,
\begin{align} \label{eq:Z_3 LMI}
\begin{bmatrix}
        Z_3 & I \\
    I & \Sigma
\end{bmatrix} \succeq 0.
\end{align}

\subsection{The joint state-input data-conforming data-driven LQR}

Toward a state-input data-conforming LQR control, we adjust Problem~\ref{prob:CE-LQR} by including the regularization term $F$ and the accompanying extra variables and LMIs defined in \eqref{eq:Z_1 LMI}, \eqref{eq:Z_2 LMI}, and \eqref{eq:Z_3 LMI}.

\begin{problem} \label{prob:State-Input Data-Conforming-CE-LQR}
Data-conforming (jointly in the state-input) data-driven LQR:
\begin{align*}
&\min_{\Sigma ,\,L,\,Z_0,\,Z_1,\,Z_2,\,Z_3} \trace \left (  Q\Sigma  \right ) + \trace \left ( R Z_0 \right) + \nonumber\\
& \hskip 5mm\gamma \Big \{ \trace \left ( \Gamma_{data}^{-1} Z_1 \right) + \trace \left ( V^{-1} Z_2 \right) + \trace \left ( \Sigma_{data} Z_3 \right) \Big \} \nonumber \\
        &\text{s.t. } \Sigma \succ 0, \quad
        \begin{bmatrix}
            Z_0 & L\\
            L^\top & \Sigma 
        \end{bmatrix}\succeq 0,\\
        &\begin{bmatrix}
        \Sigma-\widehat W-\widehat B V \widehat B^\top & \widehat A \Sigma + \widehat B L \\
        \Sigma \widehat A^\top + L^\top \widehat B^\top & \Sigma
    \end{bmatrix}\succeq 0,\\
        &\begin{bmatrix}
        Z_1 - 
        \mathcal{V}
    & \begin{bmatrix}
        \Sigma\\
        L
    \end{bmatrix}\\
    \begin{bmatrix}
        \Sigma & L^\top
    \end{bmatrix} & \Sigma
\end{bmatrix} \succeq 0, \quad
\begin{bmatrix}
        Z_3 & I \\
    I & \Sigma
\end{bmatrix} \succeq 0,\\
&\begin{bmatrix}
        Z_2 & L- H_{data}^\top \Sigma_{data}^{-1} \Sigma \\ 
        \left[L- H_{data}^\top \Sigma_{data}^{-1} \Sigma\right]^\top & \Sigma
    \end{bmatrix} \succeq 0,\\
        & \widehat A, \widehat B, \widehat W \text{ as in \eqref{eq:linear state-space ID}}.
\end{align*}
The optimal control is recovered from the optimal values $\Sigma_\star $ and $L_\star$ through $K_\star = L_\star \Sigma_\star ^{-1}$.
\end{problem}

\begin{remark} \label{remark: no zero mean condition}
\add{If the zero mean condition is to be relaxed in Assumption~\ref{assumption:Gaussianity}, the corresponding term with $\mu_{des},\mu_{data} \neq 0$ in \eqref{eq:KL divergence of Gaussians}, that is, the term
\begin{equation*}
    (\mu_{des} - \mu_{data})^\top \Gamma_{data}^{-1}(\mu_{des} - \mu_{data}),
\end{equation*}
is quadratic in $\mu_{des} = \left [\bar x^\top,\, \bar u^\top \right ]^\top$ ($\mu_{data}$ is a constant). Furthermore, the original cost \eqref{eq:costFunction} is also quadratic in $\bar x$ and $\bar u$, since $x_k$ and $u_k$ can be described by $x_k = \bar x + \tilde x_k$ and $u_k = \bar u + \tilde u_k$, where $\tilde x_k, \tilde u_k$ have zero means, and $\tilde u_k = K \tilde x_k + v_k$. Writing the cost as
\begin{align*}
    J = \lim_{k \to \infty} \E \left \{(\bar x + \tilde x_k)^\top Q (\bar x + \tilde x_k) + (\bar u + \tilde u_k)^\top R (\bar u + \tilde u_k) \right \},
\end{align*}
the cross terms, $\bar u^\top R \tilde u_k$ and $\bar x^\top Q \tilde x_k$ vanish, since $\bar x, \bar u$ are constants and $\tilde x_k, \tilde u_k, v_k$ have zero means. This splits the cost into two parts, the first
\begin{align*}
    \lim_{k \to \infty} \E \left \{\tilde x_k^\top Q  \tilde x_k + \tilde u_k^\top R \tilde u_k \right \},
\end{align*}
which results again in Problem~\ref{prob:State-Input Data-Conforming-CE-LQR}, and a deterministic second part that can be written as
\begin{equation*}
    \mu_{des}^\top 
    \begin{bmatrix}
        Q & \\
        & R
    \end{bmatrix}
    \mu_{des},
\end{equation*}
which is also quadratic in $\mu_{des}$. One can therefore solve Problem~\ref{prob:State-Input Data-Conforming-CE-LQR} for $\Sigma_\star$ (covariance scaling and rotation) then fix $K=K_\star$ and solve for $\mu_{des}$ in
\begin{align*}
    &\min_{\mu_{des}} \mu_{des}^\top 
    \begin{bmatrix}
        Q & \\
        & R
    \end{bmatrix}
    \mu_{des} + \\
    &\hskip15mm \gamma \cdot (\mu_{des} - \mu_{data})^\top \Gamma_{data}^{-1}(\mu_{des} - \mu_{data})\\
    &\text{s.t. } 
    \begin{bmatrix}
        I & 0
    \end{bmatrix}
    \mu_{des} = 
    \begin{bmatrix}
        \widehat A + \widehat B K_\star & \widehat B
    \end{bmatrix}
     \mu_{des},
\end{align*}
(density translation) which is a quadratic program with linear constraints. The two optimization problems can be solved iteratively to induce cautious translation and transformation of the state-input density.

Note\MPar{11a} that, similar to the transformation of the covariance, the translation of the mean value between operation points also implies exploring new regions of the state-input space. Our data-conforming framework limits the aggressiveness of this exploration and gives the control practitioner the time and the tools to react if harmful nonlinearities are being slowly discovered along the transition path.}
\end{remark}

\begin{remark}\label{remark:general distributions}
    \add{The\MPar{13} proposed problems and the extension in Remark~\ref{remark: no zero mean condition} can be sufficient when the linear model approximations are good enough and all distributions are characterized by their first two moments. In the general case when Assumption~\ref{assumption:Gaussianity} can no longer be sufficiently valid--the nonlinearity of the underlying dynamics and the non-Gaussianity of the distributions have to be explicitly addressed--our data-conforming framework results in a nonlinear system identification followed by a non-convex optimization problem. The modeling flexibility attained from dropping Assumption~\ref{assumption:Gaussianity} will be paid for by the increase in the computational complexity of the resulting problem.}
\end{remark}

\begin{remark} \label{remark: the ball and convex hull robustness}
As mentioned previously, the identification step, inside  Problems~\ref{prob:CE-LQR}, \ref{prob:CE-LQR-state-conforming-hard}, \ref{prob:CE-LQR-state-conforming-regularization}, and \ref{prob:State-Input Data-Conforming-CE-LQR}, resembles a naive certainty equivalence approach \cite{dean2020sample}. One can enhance the robustness of these problems by incorporating a convex hull or a ball, according to some systems' norm, centered around $(\widehat A, \widehat B)$, then search for a controller that stabilizes them all, while minimizing the cost and satisfying other constraints. This adjustment, say, to Problem~\ref{prob:State-Input Data-Conforming-CE-LQR}, can be done through more LMI constraints \cite{boyd1994linear} in the same decision variables. One also can, again, through LMI constraints,  enforce robustness under state measurement noise (state is measured through $\zeta_k$, where $\zeta_k = x_{k} + \eta_k$, for some white noise $\eta_k$) as in \cite{de2019formulas} for when the covariance of the state evolution disturbance $W \neq 0$, as in \cite{van2021matrix}.
\end{remark}

\begin{remark} \label{remark: transforming convex hull}
    The robustness results referred to in Remark~\ref{remark: the ball and convex hull robustness} hold when the underlying true system \eqref{eq:stateEquation} is linear. This motivates a new interpretation of our data-conforming framework.

    Suppose data has been collected,  a system has been identified, and a convex hull or a ball $\mathbb B$ has been constructed around this system's estimate. When the new robust (e.g., according to \cite{boyd1994linear}) controller is applied, new regions of the joint state-input space might be visited, due to the distributional shift, and new modes of nonlinearity activated. If the new system is to be identified and a new convex hull or ball to be constructed, they might be different from $\mathbb B$ used in the design process. 
    
    In other words, because of the nonlinearity of the underlying model in our case, the application of the new control law equates to the transformation of $\mathbb B$. Hence, this invalidates the premise for the quadratic stability condition, that is, that the system stays in $\mathbb B$.

    We\MPar{10} expand Problems~\ref{prob:CE-LQR-state-conforming-regularization} and \ref{prob:State-Input Data-Conforming-CE-LQR} in \cite{ramadan2024dampening} to incorporate the quadratic stability condition and show how it helps in ``dampening'' this transformation in $\mathbb B$. Therefore, improving the validity of robust and gain-scheduling design approaches \cite{boyd1994linear,van2021matrix} via solidifying the premise of invariance in $\mathbb B$ after the control application.
\end{remark}

\begin{remark}
\add{Note\MPar{2} that the proposed problems are mainly directed at the regulation problem. The tracking problem\footnote{Tracking problems or state-input-output constrained problems are typically easier to pose in a predictive (receding-horizon) fashion. An interested reader can consult \cite{ramadan2025floodgates} for predictive versions of the data-conforming approach.}, on the other hand, is more involved; since not only the control design can induce a distributional shift, but the reference trajectory to be tracked can pass through regions of the state-input space where no previous data has been collected nor a theoretical model has been derived. This can compromise safety/stability if there is a gap between the extrapolated model estimate and the true system anywhere along the reference trajectory. Therefore, non-aggressive and slowly changing reference trajectories are generally advised in the early stages of learning and data collection.}
\end{remark}

\section{On the optimality conditions of the proposed problems}\label{section:optimality conditions}
In this section we investigate the relaxations adopted in the construction of Problems~\ref{prob:CE-LQR-state-conforming-regularization} and \ref{prob:State-Input Data-Conforming-CE-LQR}, and show that the intended purpose of damping the distributional shifts is still retained after these relaxations.

Notice that Problems~\ref{prob:CE-LQR}, \ref{prob:CE-LQR-state-conforming-regularization} and \ref{prob:State-Input Data-Conforming-CE-LQR}, before the change of variables and before applying the equivalent LMI formulations, can be put in the following form (we use $A,B,W$ instead of $\widehat A,\widehat B,\widehat W$ for ease of notation)
\begin{align}
        &\min_{\Sigma ,\,K} \trace \left (  Q\Sigma  \right ) + \trace \left (  R K \Sigma K^\top  \right ) + \Xi\left (\Sigma, K \right)\nonumber\\
        &\text{s.t. }\nonumber\\
        &\Sigma \succeq \left [ A+BK\right ] \Sigma \left [ A+BK\right ]^\top + W + BV B^\top, \label{eq:Lyapunov inequality constraint}
\end{align}
where $\Xi\left (\Sigma, K \right)=0$ in Problem~\ref{prob:CE-LQR}, $$\Xi \left(\Sigma, K \right)=\gamma' \trace \left ( \left[\Sigma  - \Sigma_{data} \right]\left[\Sigma  - \Sigma_{data} \right]^\top \right)$$ in Problem~\ref{prob:CE-LQR-state-conforming-regularization}, and
\begin{align*}
    &\Xi\left (\Sigma, K \right) =\gamma \Big ( \trace \left (\Sigma_{data} \Sigma^{-1} \right ) + \trace \left (V^{-1} K \Sigma_{data} K^\top \right ) \nonumber\\
   &\hskip10mm + \trace \left (-2K^\top V^{-1}H^\top_{data}  \right ) \\
   &+ \trace \left( \mathcal{A} \Sigma + \mathcal{B}K \Sigma \right) + \trace \left( \mathcal{B}^\top \Sigma K^\top + \mathcal{C}K \Sigma K^\top \right) \Big )
\end{align*}
(up to an additive constant in $K,\Sigma$) in Problem~\ref{prob:State-Input Data-Conforming-CE-LQR}, where the matrices in calligraphic font denote the block matrices of the appropriate size of $\Gamma_{data}^{-1}$, that is,
\begin{align*}
\Gamma_{data}^{-1} = 
\begin{bmatrix}
    \mathcal{A} & \mathcal{B}\\
    \mathcal{B}^\top & \mathcal{C}
\end{bmatrix}.
\end{align*}

In the above problems, the Langrangian is given by
\begin{align*}
    &\mathcal{L}(\Sigma, K, \Upsilon) = \trace \left (  Q\Sigma  \right ) + \trace \left ( R K \Sigma K^\top \right) + \Xi\left (\Sigma, K \right)+\\
        &\trace \Bigg ( \Upsilon \Big [\left [ A+BK\right ] \Sigma \left [ A+BK\right ]^\top- \Sigma   +W +BV B^\top \Big] \Bigg),
\end{align*}
where $\Upsilon$ is the dual $r_x \times r_x$ matrix variable corresponding to the Lyapunov inequality constraint. From the dual feasibility and the zero gradient Karush-Kuhn-Tucker (KKT) \cite{boyd2004convex} optimality conditions (derivatives of traces of matrices \cite{petersen2008matrix}), the optimal solution $(K_\star,\Sigma_\star, \Upsilon_\star)$ has to satisfy \cite{helmberg2002semidefinite}:
\begin{align*}
    &\Upsilon_\star \succeq 0,\\
    &\frac{\partial}{\partial \Sigma}:\, Q + K_\star^\top R K_\star +\left[ A+BK_\star\right ]^\top \Upsilon_\star \left [ A+BK_\star\right ] - \Upsilon_\star\\
    &\hskip 20mm + \frac{\partial}{\partial \Sigma} \Xi\left (\Sigma_\star, K_\star \right) = 0.
\end{align*}

Notice that in the standard LQR formulation, Problem~\ref{prob:CE-LQR}, when $\Xi\left (\Sigma, K \right) = 0$, $\Upsilon_\star$ is nothing but the observability-type Gramian $P$ in \eqref{eq:Lyap observability P}, that is, $\Upsilon_\star=P \succ 0$. The positive definiteness of $\Upsilon_\star$ then implies that in the complementary slackness condition,
\begin{align*}
    \Upsilon_\star \Big [\left [ A+BK_\star\right ] \Sigma_\star \left [ A+BK_\star\right ]^\top- \Sigma_\star   +W +BV B^\top \Big] = 0,
\end{align*}
the constraint has to be active to satisfy the above condition, that is, $(\Sigma_\star, K_\star)$ satisfy equation \eqref{eq:Controllability Lyapunov},
\begin{align*}
    \left [ A+BK_\star\right ] \Sigma_\star \left [ A+BK_\star\right ]^\top- \Sigma_\star +W +BV B^\top = 0,
\end{align*}
or in other words, $\Sigma_\star$ is exactly the steady-state covariance of the system state we expect to see in future data when applying the law $u_k = K_\star x_k + v_k$.

This, however, is not generally true for Problems~\ref{prob:CE-LQR-state-conforming-regularization} and \ref{prob:State-Input Data-Conforming-CE-LQR}, where the Lyapunov inequality may be inactive in certain situations. The zero gradient condition for these two problems, respectively,
\begin{align*}
&Q + K_\star^\top R K_\star + \left [ A+BK_\star\right ]^\top \Upsilon_\star \left [ A+BK_\star\right ] - \Upsilon_\star\\
&\hskip20mm + 2\gamma' \left[  \Sigma_\star - \Sigma_{data} \right]= 0,
\end{align*}
and 
\begin{align*}
    & Q + \gamma\mathcal A + K_\star^\top \left [ R+\gamma\mathcal C \right] K_\star + \gamma K_\star^\top \mathcal B^\top + \gamma\mathcal B K_\star \\
    &\hskip 0mm - \gamma\Sigma_\star^{-1} \Sigma_{data} \Sigma_\star^{-1} + \left [ A+BK_\star\right ]^\top \Upsilon_\star \left [ A+BK_\star\right ] - \Upsilon_\star = 0.
\end{align*}
Using the complementary slackness condition again, to guarantee the satisfaction of the equality \eqref{eq:Controllability Lyapunov} in each problem, the terms
\begin{align}
    Q + K_\star^\top R K_\star + 2\gamma' \left[  \Sigma_\star - \Sigma_{data} \right], \label{eq:term 1}
\end{align}
and
\begin{equation}
\begin{aligned}
    Q + \gamma\mathcal A + K_\star^\top \left [ R+\gamma\mathcal C \right] K_\star + \gamma K_\star^\top \mathcal B^\top + \gamma\mathcal B K_\star \\ - \gamma\Sigma_\star^{-1} \Sigma_{data} \Sigma_\star^{-1}, \label{eq:term 2}
\end{aligned}
\end{equation}
have to be positive definite, which might not be the case if $\Sigma_{data}$ is significantly larger than $\Sigma_\star$. However, we show next that even if the terms \eqref{eq:term 1} and \eqref{eq:term 2} are not positive definite, it is still beneficial to have the inequality \eqref{eq:Lyapunov inequality constraint} inactive.

By the contrapositive argument, if the inequality constraint \eqref{eq:Lyapunov inequality constraint} is inactive, the terms \eqref{eq:term 1} and \eqref{eq:term 2} are either indefinite, or negative semi-definite. This can be implied when some eigenvalues of $\Sigma_{data}$ are significantly larger than some of $\Sigma_\star$'s. In such case, two important observations can be made: (i) when $\Sigma_{data}$ has larger eigenvalues than $\Sigma_{\star}$, the data show that aggressive exploration has been done along the direction of the corresponding eigenvectors, while less aggressive (conservative) control design is achieved across the same directions, (ii) suppose $\Sigma_{actual}$ is the solution of the Lyapunov equation \eqref{eq:Controllability Lyapunov} when $K=K_\star$, then $\Sigma_{actual}\preceq \Sigma_\star$, that is, the actual achieved design is more conservative in its exploration tendencies than the intended design. The inequality $\Sigma_{actual} \preceq \Sigma_\star$ holds because when the Lyapunov inequality \eqref{eq:Lyapunov inequality constraint} is inactive, it is implied that the optimal solution $\Sigma_\star$ to the SDP is not the fixed point solution of the Lyapunov equation \eqref{eq:Controllability Lyapunov}, hence, $\Sigma_\star$ is not the steady-state covariance, but rather an upper bound. That is, $\Sigma_{actual}$ is the smallest covariance that satisfies the Lyapunov inequality \eqref{eq:Lyapunov inequality constraint}. This can also be seen when writing the optimality conditions of the SDP with the cost $\trace \left ( \Sigma \right)$ and subject to the Lyapunov inequality \eqref{eq:Lyapunov inequality constraint}.

The above observations show that the case of inactive constraints corresponds to a more conservative achieved control. We show empirically in the next section that the Lyapunov inequality constraint \eqref{eq:Lyapunov inequality constraint} is mostly active in our simulations, and that our approaches still yield favorable results for nonlinear systems, over the certainty equivalence LQR case, even when the Lyapunov constraint is inactive.

\section{Numerical simulations} \label{section: Numerical}
The scalability of our proposed data-conforming formulations is self-evident; as convex SDPs with affine costs and LMI constraints, they can be scaled to handle problems with high state-input dimensions. Instead, we choose to work with a simple, yet telling problem, crafted to further clarify and illustrate our data-conforming paradigm.

Suppose the data-generating system \eqref{eq:stateEquation} is of the form\footnote{The results of this section can be reproduced using our open-source \textsc{Julia} code found at \href{https://github.com/msramada/data-conforming-control}{https://github.com/msramada/data-conforming-control}.} 
\begin{align} \label{eq:example1}
    \begin{pmatrix}
       x_{1,k+1}\\
       x_{2,k+1}
    \end{pmatrix}
    =
    x_{k+1} = 
    \begin{bmatrix}
        .98& .1\\
        0& .95
    \end{bmatrix}x_k + 
    \begin{bmatrix}
        0\\
        .1
    \end{bmatrix}u_k +
    w_k,
\end{align}
with the noise $w_k$ having a covariance $W = \text{diag}(0.4,\,0.1)$. \add{This\MPar{17} value of $W$ is used in simulation and data collection phases, but considered unknown in the control design phase and is to be estimated as in \eqref{eq:linear state-space ID}.} Toward solving Problems~\ref{prob:CE-LQR} and \ref{prob:CE-LQR-state-conforming-regularization}, an experimental simulation was conducted using $N=500$ collected state/input samples\MPar{3c}\footnote{\add{The generalized variances of the data covariance estimates have variances of $\mathcal{O}(N^{-1})$ (see the discussion following \eqref{eq:Sigma_data}).}}. The control $u_k = \kappa_0(x_k) + v_k$, where $\kappa_0(x_k)=K_0 x_k=\left [ -0.2,\,-9.0\right ] x_k$ is stabilizing, and $v_k$ is a white noise of zero mean and covariance $V = 0.5$ (meeting the PE condition, Assumption~\ref{assumption:PE}).

Using these simulation samples to form the data matrices \eqref{eq:data matrix}, we then solve Problems~\ref{prob:CE-LQR} and \ref{prob:CE-LQR-state-conforming-regularization} for $\gamma'=20,100,1000$. Each control gain resulting from these problems is used in a control law of the form $u_k = K x_k + v_k$ and the resulting closed-loop state distribution under each controller is shown in Figure~\ref{fig:Prob5Example}. Notice that as $\gamma'$ increases, the state design distribution converges to the state data distribution because data conformity is more heavily weighted in the control design.

Notice also in Figure~\ref{fig:Prob5Example} (a), that the (empirical) support of the state distribution of the new designed closed-loop system goes beyond that of the data. Suppose, for instance, that these not well-explored regions contain issues such as nonlinearities, discontinuities, or some safety or design condition violations. These issues, if not known and modeled by the operator in the control design, are also not discovered by the initial experiment. Hence, a non-data-conforming control, as the one in Figure~\ref{fig:Prob5Example} (a), can lead to activating these issues due to its inherent generalization beyond data. On the other hand, in (b), (c) and (d), the new controller conforms to the learning data to different levels dictated by the values of $\gamma'$ in each.

\begin{figure}
\centering 
\includegraphics[width=3.2in,height=3.6in]{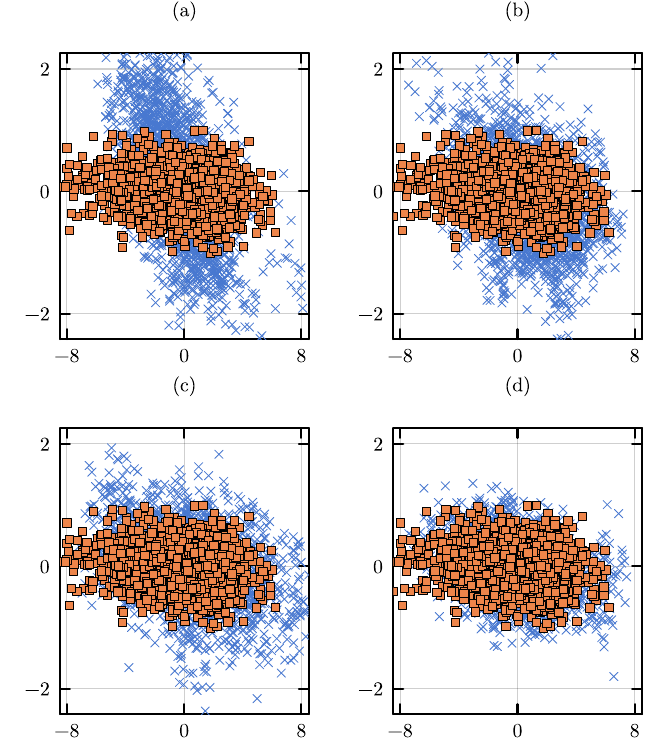} 
\caption{The x-axis and y-axis of each figure correspond to the state space, $x_{1,k}$ and $x_{2,k}$, respectively. The orange squares are the states resulting from using $K_0$ in feedback. The blue x-crosses are the states resulting from the control gains obtained by solving Problem~\ref{prob:CE-LQR} (or, equivalently, Problem~\ref{prob:CE-LQR-state-conforming-regularization} with $\gamma'=0$) in (a), and Problem~\ref{prob:CE-LQR-state-conforming-regularization} with $\gamma'=20$ in (b), $\gamma'=100$ in (c), and $\gamma'=1000$ in (d). Each empirical distribution is represented by $2000$ particles.}\label{fig:Prob5Example}
\end{figure}

To better illustrate the effect of unknown/unmodeled nonlinearities, we adjust the dynamics \eqref{eq:example1} such that it is now
\begin{align} \label{eq:example2}
    x_{k+1} = 
    \begin{bmatrix}
        .98& .1\\
        0& .95
    \end{bmatrix}x_k + 
    \begin{pmatrix}
        \theta x_{2,k}^2\\
        0
    \end{pmatrix}
    +
    \begin{bmatrix}
        0\\
        .1
    \end{bmatrix}u_k +
    w_k,
\end{align}
where $\theta=1/9$. Notice now that the nonlinearity $\theta x_{2,k}^2$ is relatively larger in the not well-explored regions. Naively applying Problem~\ref{prob:CE-LQR} can lead to  unexpected behavior and possibly to instability.

Again, we simulate the new nonlinear system with an equivalent control $u_k = K_0 x_k + v_k$ as in the previous case and for $N=500$ time steps. Then, at the last time step, we evaluate the new control laws, using the recorded $500$ data points, and apply each control law immediately in feedback for a new $2000$ time steps. Instances of the results of each simulation are illustrated in Figure~\ref{fig:Prob5Example_nonlinear}. The instability resulting from relying on Problem~\ref{prob:CE-LQR} (or, equivalently, Problem~\ref{prob:CE-LQR-state-conforming-regularization} with $\gamma'=0$) is shown in Figure~\ref{fig:Prob5Example_nonlinear}: (1). In contrast, Figure~\ref{fig:Prob5Example_nonlinear}: (2), (3), and (4) correspond to Problem~\ref{prob:CE-LQR-state-conforming-regularization} with $\gamma'=20,\,100,\,1000$, respectively. The latter three cases enforce conforming to data and hence avoid unexplored regions, in particular those where the nonlinearity is dominant and dangerous.

\begin{figure}
\centering 
\includegraphics[width=3.2in,height=3.2in]{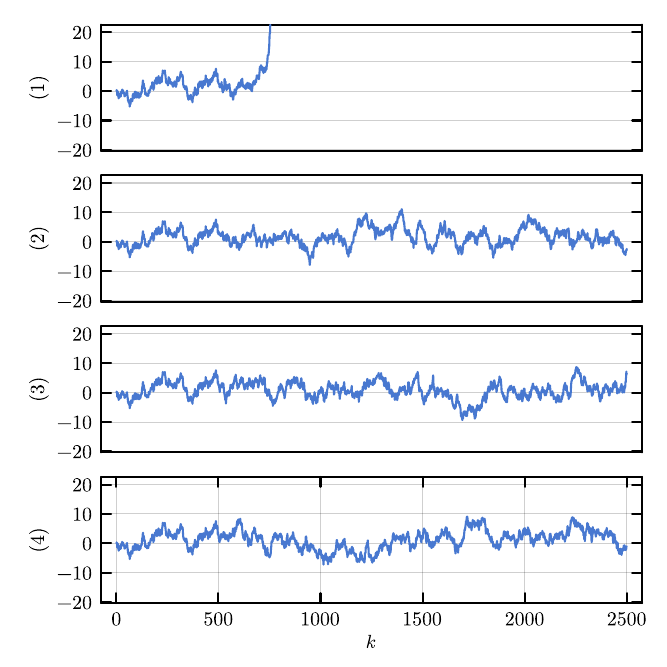} 
\caption{Values of $x_{1,k}$ for the $N=500$ data collection in the initial experiment then the $2000$ time steps under the control law resulting from solving (1) Problem~\ref{prob:CE-LQR} (or, equivalently, Problem~\ref{prob:CE-LQR-state-conforming-regularization} with $\gamma'=0$), and Problem~\ref{prob:CE-LQR-state-conforming-regularization} with (2) $\gamma'=20$, (3) $\gamma'=100$, and (4) $\gamma'=1000$.}\label{fig:Prob5Example_nonlinear}
\end{figure}

The steps of the above procedure, from (i) an initial experiment of $N=500$ bounded data, to (ii) the control design using the above four problems (Problems~\ref{prob:CE-LQR}, and Problem~\ref{prob:CE-LQR-state-conforming-regularization} with $\gamma'=20,\,100,\,1000$) and then (iii) apply each control in feedback for $2000$ time steps, are repeated ((i)--(iii)) for $1000$ repetitions. Then, we calculate the percentage of the stable simulations resulting from each procedure. The stability of each simulation is decided by the boundedness within some high threshold ($50$) in each state coordinate and pointwise in the $2000$ time steps of step (iii). The Lyapunov inequality constraint \eqref{eq:Lyapunov inequality constraint} was active in almost all of these simulations. The percentages of stable simulations are shown in Table~1, further highlighting the importance of data-conforming control design. \add{A\MPar{15a} control gain instances from these simulations are : $K_{Prob.~5/LQR}(\gamma'=0)=\left [-0.50,\,-2.03 \right]$, $K_{Prob.~5}(\gamma'=20)=\left [-0.18,\,-2.31 \right]$, $K_{Prob.~5}(\gamma'=100)=\left [-0.18,\,-3.05 \right]$, and $K_{Prob.~5}(\gamma'=1000)=\left [-0.21,\,-4.42 \right]$.}

\begin{table}[h]
\begin{center}
\begin{tabular}{ |c|c|c|c| } 
 \hline
 \multicolumn{4}{|c|}{Table 1: Percentages of stable simulations (out of 1,000 simulations)} \\
 \hline
 Problem~\ref{prob:CE-LQR}& Problem~\ref{prob:CE-LQR-state-conforming-regularization} & Problem~\ref{prob:CE-LQR-state-conforming-regularization} & Problem~\ref{prob:CE-LQR-state-conforming-regularization}\\
 (Prob.~\ref{prob:CE-LQR-state-conforming-regularization}, $\gamma'=0$) & ($\gamma'=20$)    &($\gamma'=100$) & ($\gamma'=1000$)\\
  \hline
  23.4\% & 83.7\% & 93.2\% & 98.5\% \\
 \hline
\end{tabular}
\end{center}
\end{table}

The examples in \eqref{eq:example1} and \eqref{eq:example2} share, in their structure, that the input enters the system linearly and with constant coefficients. Therefore, the effect of the input is the same on the state evolution, regardless of the contemporaneous state of the system. This is not true in general, and the effect of the input may be coupled with the current state---for instance, in bilinear systems. The following is a modification of \eqref{eq:example2} with a coupled state-input term:
\begin{equation}
\begin{aligned} \label{eq:example3}
    x_{k+1} &= 
    \begin{bmatrix}
        .98& .1\\
        0& .95
    \end{bmatrix}x_k + 
    \begin{pmatrix}
        \theta x_{2,k}^2\\
        0
    \end{pmatrix}
    + \\
    &\hskip15mm\begin{bmatrix}
        0\\
        0.1 + \theta \tanh{x_{1,k}}
    \end{bmatrix}u_k +
    w_k.
\end{aligned}
\end{equation}

Similar to the previous two examples, we run an experimental simulation of this system for $N=500$ time steps and collect the input and state data as in \eqref{eq:data matrix}. At the final time step $k=N=500$, we use the collected data in solving Problems~\ref{prob:CE-LQR}, \ref{prob:CE-LQR-state-conforming-regularization}, and \ref{prob:State-Input Data-Conforming-CE-LQR}, then use the resulting control laws for another $2000$ samples. The results of these simulations are shown in Figure~\ref{fig:Prob6Example}, showing the importance of introducing Problem~\ref{prob:State-Input Data-Conforming-CE-LQR} to account for the coupled state-input effect.

Similar to the construction of Table~1, we run the above procedure for $1000$ repetitions of bounded experiments and calculate the percentages of stable simulations resulting from each of the four control design procedures (Problem~\ref{prob:CE-LQR}, Problem~\ref{prob:CE-LQR-state-conforming-regularization} (with $\gamma'=20,\,1000$) and Problem~\ref{prob:State-Input Data-Conforming-CE-LQR} (with $\gamma=10$)). The results are recorded in Table~2, which shows the importance of using the joint state-input data distribution when the system's nonlinearities contain a coupled state-input effect, even though the Lyapunov inequality constraint \eqref{eq:Lyapunov inequality constraint} was mostly inactive. \add{A\MPar{15b} control gain instances from these simulations are : $K_{Prob.~5/LQR}(\gamma'=0)=\left [-0.63,\,-1.73 \right]$, $K_{Prob.~5}(\gamma'=20)=\left [-0.17,\,-1.16 \right]$, $K_{Prob.~5}(\gamma'=1000)=\left [-0.17,\,-1.15 \right]$, and $K_{Prob.~6}(\gamma=10)=\left [-0.19,\,-8.5 \right]$. Notice that $K_{Prob.~6}(\gamma=10)$ is close to the initial control $K_0$. 

Notice that the increase of $\gamma'$ in the previous example and $\gamma$ in the above one produces a controller closer to $K_0$. A\MPar{16} fundamental difference, however, between the designed controllers and $K_0$, is that the designed controllers are functions of $\gamma$ or $\gamma'$, which control the importance assigned to the distributional shifts penalty. By tuning $\gamma$ or $\gamma'$, the new controller is adjusted in a way respecting our distributional shift preferences, allowing for cautious and systematic change in the closed-loop behavior. On the other hand, $K_0$ is a constant (and in a general scenario, $\kappa_0$ may require intervention and strict experimental settings, as in Remark~\ref{remark:kappa_0}).}

\begin{table}[h]
\begin{center}
\begin{tabular}{ |c|c|c|c| } 
 \hline
 \multicolumn{4}{|c|}{Table 2: Percentages of stable simulations (out of 1000 simulations)} \\
 \hline
 Problem~\ref{prob:CE-LQR}& Problem~\ref{prob:CE-LQR-state-conforming-regularization} & Problem~\ref{prob:CE-LQR-state-conforming-regularization} & Problem~\ref{prob:State-Input Data-Conforming-CE-LQR}\\
 (Prob.~\ref{prob:CE-LQR-state-conforming-regularization}, $\gamma'=0$) & ($\gamma'=20$)    &($\gamma'=1000$) & ($\gamma=10$)\\
  \hline
  15.1\% & 45.5\% & 43.9\% & 99.7\% \\
 \hline
\end{tabular}
\end{center}
\end{table}

We note that in all of these simulations, none of the tested problems returned any feasibility issues. Moreover, the computation time of each problem (about six milliseconds for Problem~\ref{prob:State-Input Data-Conforming-CE-LQR} on an Apple M1 Max MacBook Pro with 64 GB or RAMs) \add{is\MPar{9c} comparable to solving a standard LQR problem (both have computational complexity $\mathcal{O}(r_x^3)$).}

\add{{Introducing\MPar{14} regularization terms to optimization problems often necetates a discussion of the sub-optimality with respect to the original cost. It is important to emphasize that in our context, supported by Tables~1 and 2, the introduction of the regularization terms make it more likely for the achieved cost to exist, by reducing the chances of instability (which leads to infinite cost). That is, with the data-conforming regularization, the achieved cost (resulting from connecting the controller in feedback to the true underlying system) is more likely to be similar to the designed cost (the cost seen by the optimization problem: resulting from connecting the controller in feedback to the approximate linear model). This is due to the fact that the data-conforming control improves the consistency between the behavior of the underlying system in closed-loop and the approximate linear model. This is different for the regularization-free case, where the design cost is finite, but the achieved one is most likely infinite (the closed-loop system is mostly unstable), as shown by Tables~1 and 2.}}

\add{The above examples show the role that data-conforming control can play in enhancing the safety of data-driven approaches. The hyper-parameters $\gamma,\gamma'$ resemble the exploration vs exploitation preferences. Starting\MPar{7} from an identified linear model that is 'sufficiently' valid for the collected data, exploitation is when our data-conforming approaches with high $\gamma,\gamma'$ result in new data that the identified model is still highly valid for (and hence the controller, since it is model-based). Exploration on the other hand, achieved by the reduction of $\gamma,\gamma'$, allows for distributional shifts beyond past data, opening new possibilities and enriching the observations we have about the underlying system.}

\add{In\MPar{11b} practice, and in the absence of accurate physical models or extensive data sets, a control engineer can start with high values of $\gamma, \gamma'$ then gradually reduce them, if needed, to allow some exploration beyond available data. One of the important recommendations of the adaptive control literature \cite{astrom1985commentary,anderson2005failures} is to model as much as possible, and leave any model reduction or abstraction for later. Our data-conforming framework can provide the safety and the time necessary for the control engineer to collect new data, model, and react to any potential activation of unknown/unmodeled nonlinearities.}

\begin{figure}
\centering 
\includegraphics[width=3.2in,height=1.6in]{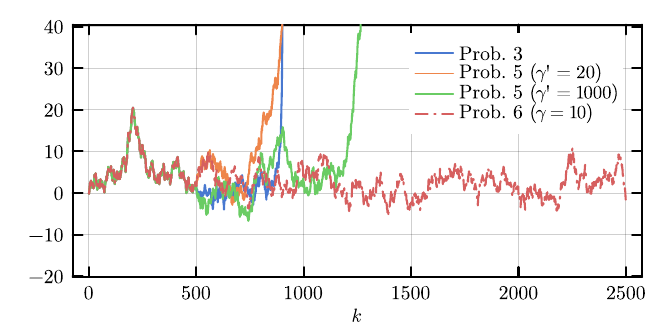} 
\caption{Value of $x_{1,k}$ for the $N=500$ samples of the initial experiment then for $2000$ time steps under the control law resulting from solving Problem~\ref{prob:CE-LQR},  Problem~\ref{prob:CE-LQR-state-conforming-regularization} with $\gamma'=20,\,1000$, and Problem~\ref{prob:State-Input Data-Conforming-CE-LQR} ($\gamma=10$). The last results in the only stabilizing control since it accounts for the joint state-input data.}\label{fig:Prob6Example}
\end{figure}

\section{Conclusion} \label{section: Conclusion}
This paper addresses a problem inherent in many modern data-driven and adaptive control approaches, namely, the problem of the premature and possibly false generalization beyond data, together with its consequences in terms of sudden distributional shifts in the state-input space. We present methods for mitigating this problem through enforcing consistency with the learning data, and we numerically test them. Because of  the formulation of our methods as solutions to SDPs of affine costs and LMI constraints, they are computationally efficient \add{(comparable\MPar{9b} computational complexity to regular LQR)}, can scale up to systems with hundreds in dimension, and can be easily integrated with modern control design approaches.

Further work is being pursued to \add{(i) employ recent techniques \cite{de2023learning} that have the potential to further improve the validity of the linear identification step (ii) relax\MPar{6b} Assumption~\ref{assumption:Gaussianity} and allow for identification using nonlinear parametrizations and regularization in terms of more flexible distributions beyond Gaussian (in the sense of Remark~\ref{remark:general distributions})} (iii) investigate the possibility of developing data-conforming policy gradient surrogates to dampen distributional shifts resulting from standard policy gradient methods.
\section*{Acknowledgment}
This material was based upon work
supported by the U.S. Department of Energy, Office of Science,
Office of Advanced Scientific Computing Research (ASCR) under
Contract DE-AC02-06CH11347. 

\balance
\bibliographystyle{IEEEtran}
\bibliography{References}
\vspace{0.1cm}
\begin{flushright}
	\scriptsize \framebox{\parbox{2.5in}{Government License: The
			submitted manuscript has been created by UChicago Argonne,
			LLC, Operator of Argonne National Laboratory (``Argonne").
			Argonne, a U.S. Department of Energy Office of Science
			laboratory, is operated under Contract
			No. DE-AC02-06CH11357.  The U.S. Government retains for
			itself, and others acting on its behalf, a paid-up
			nonexclusive, irrevocable worldwide license in said
			article to reproduce, prepare derivative works, distribute
			copies to the public, and perform publicly and display
			publicly, by or on behalf of the Government. The Department of Energy will provide public access to these results of federally sponsored research in accordance with the DOE Public Access Plan. http://energy.gov/downloads/doe-public-access-plan. }}
	\normalsize
\end{flushright}	

\begin{IEEEbiography}[{\includegraphics[width=1in,height=1.25in,clip,keepaspectratio]{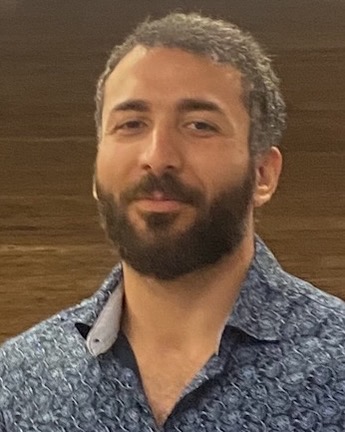}}]{Mohammad S. Ramadan} (Member, IEEE) is a postdoctoral appointee in the Mathematics and Computer Science Division at Argonne National Laboratory. He obtained his B.S. degree in aeronautical engineering from Jordan University of Science and Technology in 2016 and his Ph.D. in mechanical and aerospace engineering from University of California San Diego in 2023. His research interests include stochastic optimal control, state estimation, and data-driven robust control.
\end{IEEEbiography}

\begin{IEEEbiography}[{\includegraphics[width=1in,height=1.25in,clip,keepaspectratio]{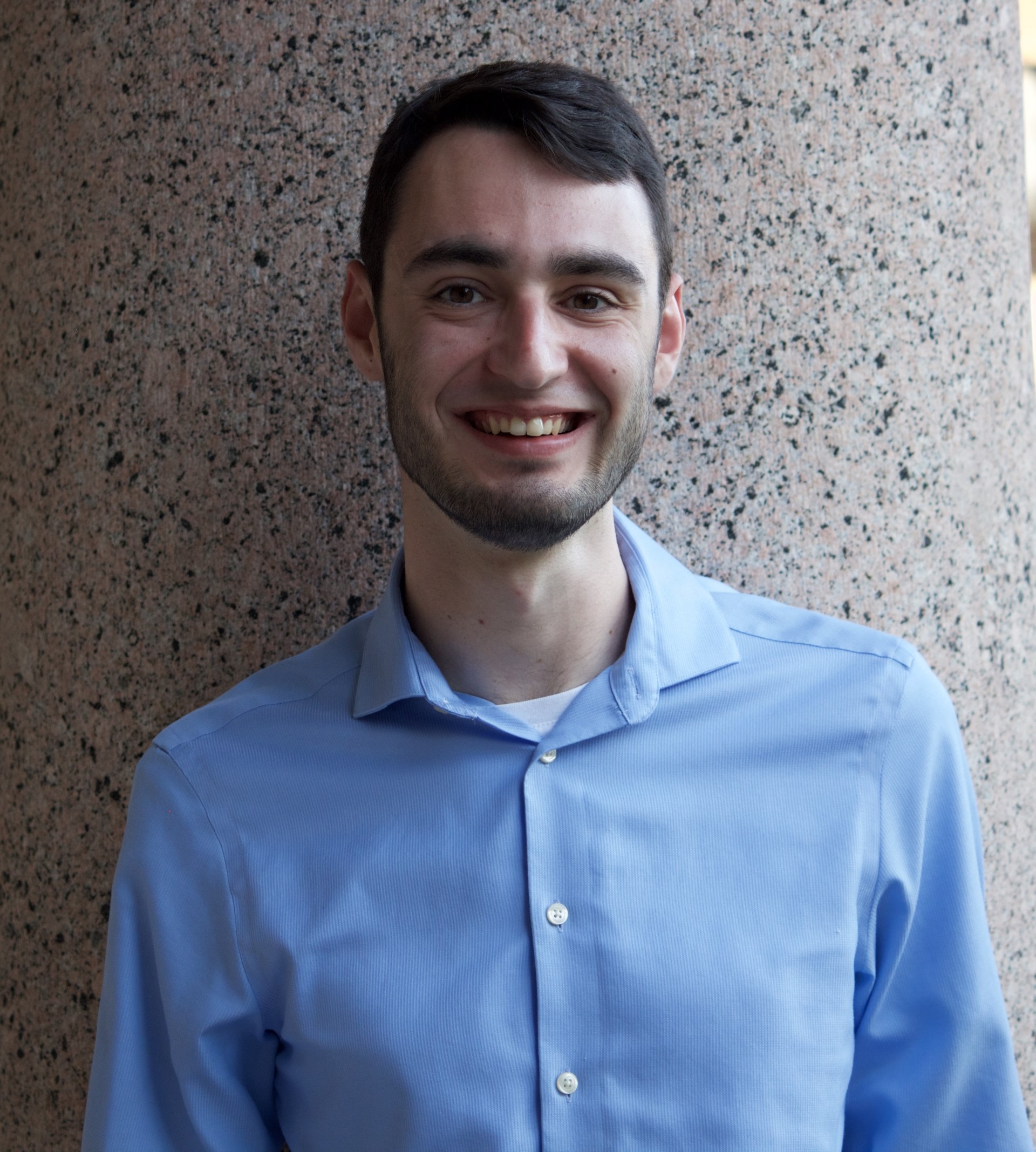}}]{Evan Toler} is a postdoctoral appointee in the Mathematics and Computer Science Division at Argonne National Laboratory. He received his Ph.D. in mathematics from the Courant Institute of Mathematical Sciences at New York University, where he was supported by a National Science Foundation Graduate Student Fellowship. He received his B.A. in computational and applied math and statistics from Rice University. His research interests are in numerical methods and optimization for systems governed by partial differential equations. He is particularly interested in applications in magnetic confinement fusion.
\end{IEEEbiography}

\begin{IEEEbiography}[{\includegraphics[width=1in,height=1.25in,clip,keepaspectratio]{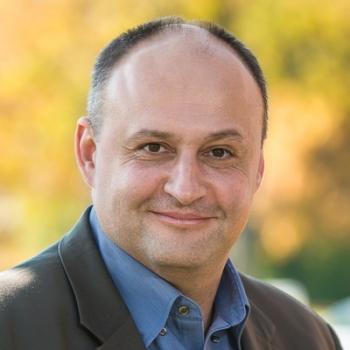}}]{Mihai Anitescu} (Member, IEEE) is a senior computational mathematician in the Mathematics and Computer Science Division at Argonne National Laboratory and a professor in the Department of Statistics at the University of Chicago. He obtained his engineer diploma (electrical engineering) from the Polytechnic University of Bucharest in 1992 and his Ph.D. in applied mathematical and computational sciences from the University of Iowa in 1997. He specializes in
the areas of numerical optimization, computational science, numerical analysis and uncertainty quantification in which he has published more than 100 papers in scholarly journals and book chapters. He has been recognized for his work in applied mathematics by his selection as a SIAM Fellow in 2019.
\end{IEEEbiography}

\appendix
\section{Appendix}

\subsection{The two descriptions of the cost $J$}
\label{Appendix:cost descriptions: sum and limit}
\add{We\MPar{5c} assume that $x_k$ and $u_k$ have bounded first two moments (possibly under a stabilizing feedback control law) for all $k=0,1,\hdots,$ and that they converge to stationary processes in steady-state ($k \to \infty$). Starting from the first description (for compactness, let $\ell(x_k, u_k) = x_k^\top Q x_k + u_k^\top R u_k$),
\begin{align*}
    J &= \lim_{T \to \infty} \E \left \{ \frac{1}{T} \sum_{k=0}^T \ell(x_k, u_k) \right \},\\
    &= \lim_{T \to \infty}  \frac{1}{T} \sum_{k=0}^T \E \ell(x_k, u_k) ,\\
    &= \lim_{T \to \infty} \frac{1}{T} \Big \{ \sum_{k'=0}^{T_0-1} \E \ell(x_{k'}, u_{k'})  +  \sum_{k=T_0}^{T} \E  \ell(x_k, u_k) \Big \},
\end{align*}
and the first term vanishes for any $T_0$, so we have
\begin{align*}
    J &= \lim_{T \to \infty} \frac{1}{T}  \sum_{k=T_0}^{T} \E \left \{ \ell(x_k, u_k) \right \}.
\end{align*}
Using the boundedness and asymptotic stationarity assumptions above, for any $\epsilon >0$, there is a $T_0$ such that $| \E \ell(x_{T_0}, u_{T_0}) -  \E \ell(x_k, u_k) | < \epsilon$, for all $k \geq T_0$, and therefore,
\begin{align*}
    &|\frac{T-T_0}{T}\E \ell(x_{T_0}, u_{T_0}) -\frac{1}{T}  \sum_{k=T_0}^{T} \E \ell(x_k, u_k)|\\
    &\hskip 5mm \leq \frac{1}{T}  \sum_{k=T_0}^{T} |\E \ell(x_{T_0}, u_{T_0})- \E  \ell(x_k, u_k)| < \frac{T-T_0}{T} \epsilon \leq \epsilon.
\end{align*}
This holds for any $T \geq T_0$, so
\begin{align*}
    &|\lim_{T \to \infty} \frac{T-T_0}{T}\E \ell(x_{T_0}, u_{T_0}) -J|=  |\E \ell(x_{T_0}, u_{T_0}) -J| < \epsilon.
\end{align*}
And $\epsilon$ can be made arbitrarily small, therefore $J = \lim_{T_0 \to \infty} \E \ell(x_{T_0}, u_{T_0})$, which is the second description in \eqref{eq:costFunction}.}

\subsection{The cost in the $P-$parametrization} \label{Appendix:P parametrization cost}
Starting from on the descriptions of the cost \eqref{eq:costFunction}, say
\begin{align*}
    J = \lim_{k \to \infty} \E \left \{x_k^\top Q x_k + u_k^\top R u_k \right \},
\end{align*}
and after substituting $u_k = K x_k +v_k$ \add{($v_k$ has zero mean, covariance $V$, and is independent of $x_k$)},
\begin{align*}
    J = \lim_{k \to \infty} \E \left \{x_k^\top \left [ Q + K^\top R K \right ] x_k + v_k^\top R v_k \right \}.
\end{align*}
Using the cyclic property of the trace and the linearity of $\E$, $\E \left \{v_k^\top R v_k \right \} = \trace \left ( R V \right )$, which is an additive constant (for a fixed $V$) and can be omitted from the cost. Now, using the cyclic property of the trace and the linearity of $\E$ again, we have
\begin{align}
    J &\propto \lim_{k \to \infty} \E \left \{x_k^\top \left [ Q + K^\top R K \right ] x_k \right \},\nonumber\\
    &=  \lim_{k \to \infty} \E \trace \left ( \left [ Q + K^\top R K \right ] x_k x_k^\top \right ),\nonumber\\
    &= \trace \left( \left [ Q + K^\top R K \right ] \lim_{k \to \infty} \E \left \{x_k x_k^\top \right \} \right),\nonumber\\
    &= \trace \left ( \left [ Q + K^\top R K \right ] \Sigma \right). \label{eq:App:cost description with Sigma}
\end{align}
The steady-state covariance of the state is designated by $\Sigma = \lim_{k \to \infty} \E \left \{x_k x_k^\top \right \}$, which is also known as the controllability-type Gramian. Since
\begin{align*}
    x_k = \left [ A + B K \right ]^k x_0 + \sum_{i=0}^{k-1} \left[ A+BK\right]^{k-1-i} \left ( w_i + B v_i \right ),
\end{align*}
and using the fact that $v_i,\,w_i$ are white, of zero mean and mutually independent from each other and from $x_0$, this covariance can be described by
\begin{align*}
    \Sigma &= \lim_{k \to \infty} \Big \{ \left [ A + B K \right ]^k \Sigma_{x_0} \left [ A + B K \right ]^{k,\,\top} + \\
    &\hskip-5mm\sum_{i=0}^{k-1}  \left[ A+BK\right]^{k-1-i} \left [ W + B V B^\top \right ]\left[ A+BK\right]^{k-1-i,\,\top}\Big \},
\end{align*}
or, if $A+BK$ is hurwitz, the transient term vanishes and we have
\begin{align*}
    \Sigma = \sum_{k=0}^{\infty}  \left[ A+BK\right]^{k} \left [ W + B V B^\top \right ]\left[ A+BK\right]^{k,\,\top}.
\end{align*}
\add{This completes the derivation of the cost under the controllability-type Gramian parametrization \eqref{eq:cost Controllability type}. For the observability-type one, after substituting this description of $\Sigma$ above in \eqref{eq:App:cost description with Sigma}, and using the cyclic property of the trace operator, we have (up to an additive constant)}
\begin{align*}
    J &= \trace \left ( \left [ Q + K^\top R K \right ] \Sigma \right) \\
    &=\trace \Big ( \left [ Q + K^\top R K \right ] \times \\
    &\hskip 5mm\sum_{k=0}^{\infty}  \left[ A+BK\right]^{k} \left [ W + B V B^\top \right ]\left[ A+BK\right]^{k,\,\top} \Big),\\
    &= \trace \Big ( \sum_{k=0}^{\infty} \left[ A+BK\right]^{k,\,\top} \left [ Q + K^\top R K \right ]   \left[ A+BK\right]^{k} \times \\
    &\hskip5mm\left [ W + B V B^\top \right ] \Big),\\
    &= \trace \left ( P \left [ W + B V B^\top \right ] \right),
\end{align*}
where $P$ is the observability-type Gramian, given by
\begin{align*}
    P = \sum_{k=0}^\infty  [A+BK]^{k\,\top} \left [ Q + K^\top R K \right ] [A+BK]^k.
\end{align*}
\end{document}